\documentclass[12pt]{article}
\usepackage[right=1in,left=1in,top=1in,bottom=1in]{geometry}
\usepackage{hyperref}
\hypersetup{colorlinks, citecolor=blue, filecolor=blue, linkcolor=blue, urlcolor=blue}
\usepackage{graphicx}
\usepackage{url}
\usepackage[round]{natbib}
\usepackage{amsmath,amsthm} 
\usepackage{engord}
\usepackage{float, bbm}
\usepackage{subfig}
\usepackage{tikz}
\usetikzlibrary{arrows.meta, positioning}
\usepackage{pdflscape}
\usepackage{booktabs}
\usepackage{pgfplots, verbatim}
\pgfplotsset{compat=1.14}
\pgfplotsset{every axis label/.append style={font=\tiny}}
\usepackage[labelsep=period]{caption} %% This switches "Table 1: Title" to "Table 1. Title"

\usepackage{amssymb} %% Necessary, just for the \checkmark command  in tables.
\usepackage{multirow, authblk} %% Necessary if we are doing tables in LaTeX

\usepackage{xr}

\usepackage{setspace}
\usepackage{sectsty}
\usepackage{bbm}
\usepackage{natbib}

\usepackage{hyperref}
\hypersetup{colorlinks, linkcolor=blue, urlcolor=blue, citecolor=blue}

\title{Pre-auction strategic communication}

\author{Eric Yan\thanks{Massachusetts Institute of Technology, eyan@mit.edu.} \footnote{I thank Drew Fudenberg, Ellen Muir, and Alex Wolitzky for generous guidance and suggestions. For helpful comments and discussion I thank Ian Ball, Abhijit Banerjee, Alessandro Bonatti, Lucas Barros, Yifan Dai, Glenn Ellison, Bob Gibbons, Andrew Koh, Elliot Lipnowski, Daniel Luo, Eric Maskin, Anna Merotto, Jan Morgan, Lea Nagel, Ziwen Sun, Otávio Tecchio, Jean Tirole, and attendees of MIT Theory Lunch and 14.192, Sloan Organizational Lunch, and Harvard Econ 2049. Refine.ink was used to proofread this paper.}}
\date{\today}

\newtheorem{assumption}{Assumption}
\newtheorem{definition}{Definition}
\newtheorem{proposition}{Proposition}
\newtheorem{corollary}{Corollary}

\newtheorem{remark}{Remark}
\newtheorem{lemma}{Lemma}

\newtheorem{theorem}{Theorem}

\def\Rev{\mathrm{Rev}}
\def\E{\mathbb{E}}
\def\supp{\mathrm{supp}}
\begin{document}

\maketitle

\begin{abstract}
    \noindent High-stakes auctions are often preceded by nonbinding communication between bidders and the seller. I study a two-period model in which bidders privately send cheap-talk messages to the seller about their valuations before the seller decides whether to run a mechanism or take a forfeitable outside option. The seller has commitment within the second-period mechanism but not over how she uses first-period communication. In monotone bidder-symmetric equilibria, an unrestricted seller runs at most one mechanism on the equilibrium path: a second-price auction with a common reserve, and only when all bidders report values above a common threshold. Thus discriminatory auctions cannot arise in equilibrium, despite asymmetric on-path posteriors. With two bidders, the seller is better off, and can sometimes attain her full-commitment payoff, if she can commit ex-ante to second-price auctions with a common reserve.
\end{abstract}

\clearpage
\section{Introduction}
Standard auction theory takes the seller's decision to sell and the information used to set allocation rules as given. In practice,  many high-stakes auctions are preceded by completely nonbinding bidder communication — such as indications of interest or requests for comment — that shapes whether the auction is run at all and what reserve prices are chosen. For example, in corporate mergers and acquisitions, a target firm typically receives indications of interest from potential acquirers before deciding whether to initiate a formal sale process.\footnote{\cite{ye_indicative_2007} and \cite{quint_theory_2018} model aspects of these two-stage processes.} Regulations have formalized similar practices in other contexts as well: Since 2019 in the U.S., the SEC has allowed all companies to engage in ``testing-the-waters" communications with potential investors to evaluate market interest before deciding whether or not to run an IPO, describing such communications as a ``cost-effective means for evaluating market interest\dots before incurring the costs associated with such an offering."\footnote{The original law allowing such practices was the Jumpstart Our Business Startups Act, Pub. L. No. 112-106, 126 Stat. 306 (2012). SEC Rule 163B extended the law to all companies. See the SEC \href{https://www.sec.gov/resources-small-businesses/small-business-compliance-guides/solicitations-interest-prior-registered-public-offering}{website} for a description.} Additionally, prior to running spectrum auctions, the FCC seeks public comment, expressly stating that such communication will be used to determine spectrum clearing targets and opening bid prices.\footnote{See the Federal Register Vol. 80, No. 19 for a \href{https://www.govinfo.gov/content/pkg/FR-2015-01-29/pdf/2015-01607.pdf}{sample} request. While the full spectrum auction process is quite complicated, the key seller considerations in this paper — whether to put an object up for sale and reserve pricing — seem to be relevant in this context as well.} The widespread use of these practices suggests that sellers in these markets find such communications informative. This paper formulates a simple model of pre-auction communication, which is taken as part of an extensive-form game where the seller and bidders do not have commitment over how any information transmitted is used or disclosed, and analyzes what kind of information can be transmitted in equilibrium.

Specifically, this paper examines a two-period model in which bidders (referred to by he) send private cheap talk messages to the seller (she) about their valuations, and the seller decides in the second period whether to run a mechanism or take an outside option that disappears if she chooses to run the auction. This outside option can represent, for example, a negotiated sale of a company to a known buyer or broker.\footnote{This formulation is equivalent to one in which the seller needs to pay a sunk cost in order to run the sale process, which in some applications is also a natural interpretation. In particular, takeover auctions in mergers and acquisitions often require the costly investment banking and lawyer fees as part of the due diligence process.} The seller has \emph{limited commitment}: the seller has commitment within the second period if she chooses to run a mechanism but does not have between-period commitment, so she can use any information she collects in the first period to determine whether to run a mechanism in the second period and what the allocation rule should be. As such, the bidders trade off providing the seller enough information to induce a sale with retaining their information rents in the eventual mechanism the seller chooses.

This paper studies bidder-symmetric (henceforth symmetric) perfect Bayesian equilibria of this model under two regimes. In the first regime, the seller can choose any mechanism in the second period, and hence will choose an allocation rule according to \cite{myerson_optimal_1981}. A well-known feature of the optimal mechanism is that the seller will set \emph{individualized} reserve prices when she faces asymmetric posteriors. In the second regime, the seller must choose a second-price auction with a single, \emph{common} reserve price whenever she runs a mechanism. The first main result (Theorem \ref{myerson_characterization}) shows that, restricting attention to pure monotone bidder strategies, any symmetric equilibrium under the unrestricted regime is outcome-equivalent to a threshold equilibrium where each bidder simply reports whether his value exceeds a threshold, and the seller runs the auction if and only if all bidders report being above it. It follows that the seller only runs one mechanism on the equilibrium path: a second-price auction with a common reserve $r$ weakly above the seller's monopoly price under the prior. This holds even though on-path messages may induce asymmetric posteriors across bidders. These threshold equilibria are also shown to satisfy a D1-style restriction on seller beliefs off-path, ensuring that informative communication in this multi-sender cheap talk model does not depend on unreasonable behavior off-path.

At first glance, one might expect at least some variation in the mechanisms the seller runs on-path. High-value bidders might be willing to communicate more precise information, and therefore face higher expected payments, to ensure that the seller runs an auction, while lower or intermediate types might prefer noisier communication that makes the auction less likely but lowers the expected payment conditional on sale. Relatedly, one might expect the seller to want to run the auction when a subset of bidders reports being above the threshold. If the seller actually ran the auction in any of these cases, however, bidders would face an information-provision \emph{freeriding} problem. Each bidder would like the seller to receive enough favorable information to run the auction, but each bidder would rather this information come from another bidder. The basic reason is that the \cite{myerson_optimal_1981} auction favors weak bidders to extract more rents from strong bidders when the seller has asymmetric posteriors. A bidder who reveals that he is strong makes the seller more willing to run the auction, but also induces the seller to choose a mechanism that extracts more surplus from him. Therefore, when the seller is unrestricted, if a bidder is providing information that is favorable to the seller running an auction, that information must be pivotal in the seller's decision to run the auction as well. In particular, while the seller may be tempted to design mechanisms tailored to each bidder when faced with asymmetric distributions, doing so necessarily comes at the expense of bidder types who have a high value but are close to the individualized reserve they face after sending their cheap-talk message. These bidder types would rather deviate and obtain the lower reserve if that is a possible seller action in equilibrium.

A corollary of this result is that the seller's payoff in any equilibrium of this model is bounded strictly below her \emph{full-commitment} payoff, defined as her payoff if she has cross-period commitment and no restrictions on mechanisms. If the seller has full commitment, the problem is effectively static and reduces to that of \cite{myerson_optimal_1981}: Bidders report their values truthfully, and the seller commits to running the auction if and only if at least one bidder's virtual type (calculated according to the ex-ante prior) is at least the value of the seller's outside option.

On the other hand, depending on the value of the outside option, the seller can sometimes attain her commitment payoff under the second regime where she is restricted to selecting a single reserve price. The restricted regime reflects binding constraints that sellers can face in practice. Government sellers are often explicitly required to treat bidders symmetrically: For example, EU public procurement directives mandate equal treatment of tenderers.\footnote{See, e.g., Directive 2014/24/EU of the European Parliament, Article 18(1): ``Contracting authorities shall treat economic operators equally and without discrimination.''} Additionally, in certain corporate takeovers, target firms have faced lawsuits when found to have accepted a bid that is below the highest value.\footnote{Under Delaware law, the duties articulated in \emph{Revlon, Inc.\ v.\ MacAndrews \& Forbes Holdings, Inc.}, 506 A.2d 173 (Del.\ 1986), require in certain circumstances for the board to seek ``the best price" for shareholders.} Intuitively, if the seller has ex-ante commitment to running a nondiscriminatory auction format, bidders are much more easily incentivized to provide precise information that guarantees the seller runs an auction, allowing higher thresholds to be sustained in equilibrium. This comparison provides an optimality foundation for the use of simple and nondiscriminatory auction formats in practice — such as the second-price auction with a single reserve — in settings where the seller presumably has asymmetric information about bidders, if this asymmetry arises due to communications received before the mechanism has been announced.

From a technical perspective, the game analyzed is most similar to cheap talk models with multiple senders \citep{battaglini_02, takahashi_multi_08}. However, it differs from most of the literature on strategic communication in two important ways. Preferences are quasilinear — putting the payoff environment outside the analysis of \cite{crawford_strategic_1982} even if there was just one bidder — and not state-independent, as analyzed by \cite{lipnowski_cheap_2020}. In particular, even though all bidder types weakly prefer lower reserve prices, each bidder is indifferent between \emph{all} mechanisms that give him zero utility, and this set is different for each bidder type. It also makes each bidder (sender) indifferent between a large set of seller (receiver) best-responses, exacerbating the already severe multiplicity-of-equilibria problem faced by cheap talk models. As such, this paper adopts several refinements used in the literature to narrow the set of equilibria considered.

\subsection{Related Literature}

This paper contributes to the extensive literature on strategic information transmission initiated by \cite{crawford_strategic_1982}. Specifically, it relates to models of cheap talk with multiple senders, such as \cite{battaglini_02} and \cite{takahashi_multi_08}. The role of pre-play communication has been analyzed in buyer-seller interactions by \cite{farrell_cheap_1989} in the context of bargaining, and by \cite{matthews_pre-play_1989} in the context of bilateral trade. Both papers find that pre-play communication can enlarge the set of equilibria in these models.

The paper also adds to the literature on mechanism design with limited commitment. In principal-agent settings, \cite{Bester_strausz_01}, \cite{skreta_sequentially_2006}, and \cite{doval_mechanism_2022} emphasize the ``ratchet effect," where an agent's revelation of high types today leads to more extractive contracts in future periods. In the context of auctions, \cite{skreta_2015} and \cite{liu_auctions_2019} study how a seller's inability to commit to not lowering her reserve price over time decreases her payoff in equilibrium. This paper focuses on how the initial decision to run a mechanism is shaped by limited commitment in the context of pre-auction communication. Additionally, while previous literature has documented that a seller with limited commitment can improve her payoff by restricting her choice of mechanisms, existing papers focus on settings with a single agent and argue that such restrictions limit the amount of information the principal can learn about the agent's type. This paper documents that nondiscrimination, a natural source of additional ex-ante commitment in multi-agent settings, can also improve the seller's payoff. Finally, the seller's use of communication in the first period to determine allocation rules in the second also makes the model similar to \cite{Courty_li_2000}, but the seller in their model has cross-period commitment and the buyer gains information about his value between periods.

The model in this paper is conceptually related to the literature on auctions with endogenous bidder entry, including \cite{mcafee_auctions_1987} and \cite{levin_equilibrium_1994}. In those models, bidders decide whether to pay a sunk cost to learn their valuation and enter the auction, leading to an equilibrium entry threshold where the expected profit from participation equals the entry cost. This paper can be thought of as a model of auctions with endogenous seller entry.

Some papers have explicitly modeled indicative bidding in financial processes, such as the two-stage auctions common in corporate takeovers \citep{ye_indicative_2007, quint_theory_2018}.\footnote{Motivated by preplay communication in stock market exchanges, \cite{biais_2014} experimentally study the design of preopening mechanisms prior to call auctions.} These models take the stance that bidders must pay a cost to enter an auction, such as due diligence costs, and therefore may not want to enter if competition is too fierce. \cite{ye_indicative_2007} shows that there does not exist a symmetric equilibrium with increasing strategies in a model where the bidders make completely nonbinding indicative bids in the first round, and the seller selects the $n$ highest to participate in the auction and pay the entry fee. \cite{quint_theory_2018} construct symmetric equilibria by coarsening the message space of bidders and endowing the seller (receiver) with some commitment over how bidders get selected. This paper focuses on the seller's entry decision and limited commitment as a force that disciplines truth-telling in the communication stage.

The remainder of the paper is organized as follows. Section \ref{model} presents the baseline model and equilibrium concepts, including restrictions on off-path beliefs for the seller. Symmetric equilibria of the model with an unrestricted seller are characterized in section \ref{unrestricted}. Section \ref{restricted} shows that the seller does better by restricting admissible mechanisms to those with a common reserve and solves for optimal equilibria and comparative statics. Robustness of the equilibria in the baseline model is discussed in section \ref{extensions}.
\section{Model}\label{model}

There are $n\geq 2$ risk-neutral bidders and a single seller, indexed by $i\in I:=\{1,\dots,n\}$. Each bidder $i$ has a private valuation $v_i$ drawn independently from an atomless, commonly known CDF $F\in C^1([\underline v,\overline v])$ with PDF $f$, fully supported on $\Theta_i=[\underline v,\overline v]\subseteq \mathbb{R}_+$, and let $\Theta=\prod_{i=1}^n \Theta_i$. The seller has an indivisible good and an outside option with value $c\in(\underline v,\overline v)$. For expositional simplicity, I assume the distribution $F$ is regular, meaning that its virtual type function $\varphi$ is increasing, and admits a unique monopoly price $p_F$. If the seller runs an auction, the outside option is forfeited.

 At $t=1$, each bidder $i$ observes $v_i$ and sends $m_i\in M_i$ privately according to a measurable pure monotone strategy $\sigma_i:[\underline v,\overline v]\to M_i$, where $M_i$ is a totally ordered set, such as $M_i=[\underline v,\overline v]$.\footnote{The analysis also goes through if mixed strategies are allowed, and one assumes $v'>v\Rightarrow \inf\sigma_i(v')\geq \sup\sigma_i(v)$. The message space and strategy restrictions are discussed in subsection \ref{subsec:monotone}.} After receiving message $m_i$, the seller updates her beliefs and forms a posterior $F_{m_i}$ about bidder $i$'s valuation. Given $\sigma_i$, for every $m_i$ sent on path, the seller forms her posterior $F_{m_i}$ via Bayes' rule.\footnote{See subsection \ref{subsec:solution} for the formal solution concept.}

Given an on-path message profile $m=(m_1,\dots,m_n)$, the seller forms updated beliefs $(F_{m_1},\dots,F_{m_n})$. At $t=2$, the seller chooses an action in $\{\varnothing\}\cup\mathcal G$, where $\varnothing$ denotes taking the outside option, yielding payoff $c$ to the seller and $0$ to bidders, and $\mathcal G$ is the set of admissible mechanisms. A mechanism $\Gamma\in\mathcal G$ is defined as a pair $(x,t)$, where for each $i$, the function $x_i:\supp(F_{m_i})\to[0,1]$ is the allocation rule for bidder $i$, subject to the usual feasibility constraints, and $t_i:\supp(F_{m_i})\to\mathbb R$ is the transfer rule.\footnote{See Assumption \ref{ass:posterior_support} for a discussion of the domains of $x_i$ and $t_i.$}

I analyze the game under two regimes that specify the set of admissible mechanisms $\mathcal{G}$:

\begin{enumerate}
    \item \textbf{Unrestricted:} The seller can choose any allocation rule that satisfies dominant-strategy incentive-compatibility (DSIC) and interim individual rationality (IIR) with respect to her posterior beliefs. Under this regime, the seller's problem reduces to that of \cite{myerson_optimal_1981}.

    \item \textbf{Restricted:} The seller must choose a second-price auction with a common (potentially randomized) reserve price. This is the set of $(x, t)$ for which there exists a CDF $H$ supported inside $[\underline v,\overline v]$ over reserve prices, such that for all $i$ and types $v_i$:
        \[x_i(v_i)=\mathbb{E}_H\left[\Pr\left(v_i>\max_{j\neq i}v_j,\ v_i\geq r\right)\right].\]
\end{enumerate}

Bidders have quasi-linear utility functions $u_i(v)=x_i(v)v-t_i(v)$, and the seller's utility is the sum of payments if the mechanism is run, and $c$ otherwise. Bidders are assumed to play the seller-optimal equilibrium of whatever mechanism the seller chooses at $t=2$. The DSIC requirement does not restrict the space of allocation rules and payments that the seller can choose from in the second period relative to Bayesian incentive-compatibility \citep{Manelli_vincent_2010, gershkov_equivalence_2013}. Together with the assumption that bidders play the seller-optimal equilibrium at $t=2$, these restrictions control the analysis of ``double deviations," ensuring that a bidder who deviates at $t=1$ cannot compound the deviation by misreporting at $t=2$. Therefore, the equilibrium analysis can focus purely on first-period communication incentives.\footnote{The DSIC assumption can be dispensed with if one instead clarifies that each bidder observes which bidders have advanced to the auction and that bidders play the seller-optimal Bayesian Nash equilibrium of the chosen mechanism. Observing the set of participants allows each bidder to gain enough information to always play the seller-optimal equilibrium strategy. Under a DSIC implementation, however, bidders do not need to observe who else has advanced to the auction, which is natural in some applications.} Plainly, since the game technically requires bidders to report messages to the seller at both $t=1$ and $t=2$, these assumptions are used to reduce the model to a more standard cheap talk model where bidders (senders) send messages about their valuations, and the seller (receiver) responds by choosing an action that determines payoffs.

Completing the description of the game requires addressing an additional related complication: after deviating in the communication stage, a bidder may have a value outside of the support of the seller's beliefs. If the seller were to form beliefs about this bidder based on off-path reports at the mechanism stage, many such beliefs would be incompatible with equilibrium or affect communication incentives in an arbitrary way. Rather than allowing off-support reports at $t=2$ to act as an additional communication channel between the bidders and the seller, I make the following assumption.

\begin{assumption}
\label{ass:posterior_support}
If the seller runs a mechanism after receiving $m=(m_1,\dots,m_n)$, she defines bidder $i$'s allocation and payment rules $x_i(\cdot\mid m_i,m_{-i})$ and $t_i(\cdot\mid m_i,m_{-i})$ on $\supp(F_{m_i})$. If the seller runs a mechanism after the message profile $m,$ a bidder $i$ with value $v_i$ gets utility
\begin{equation*}
    u_i(v_i\mid m_i,m_{-i})
    =
    \sup_{\hat v\in\supp(F_{m_i})}
    \left\{
        v_i x_i(\hat v\mid m_i,m_{-i})
        -
        t_i(\hat v\mid m_i,m_{-i})
    \right\}
\end{equation*}
from the mechanism.
\end{assumption}

Assumption \ref{ass:posterior_support} should be interpreted as a restriction on how the seller treats a bidder after his first-period message, not as a literal restriction that a bidder is physically unable to report certain values at the mechanism stage. After observing $m_i$, the seller designs a continuation menu for the values she regards as possible under $F_{m_i}$. If a bidder deviated at $t=1$ and his true value lies outside this support, his payoff is the supremum payoff attainable from that menu.

The specific implementation given by the assumption is not essential. Note that Assumption \ref{ass:posterior_support} implicitly defines an extension of the allocation and transfer rules to the full type space $[\underline v,\overline v]$. Equivalently, one could assume that after each message profile the seller directly chooses these extended rules. The key point is that the first-period message determines the seller's treatment of the bidder in the continuation mechanism, while an out-of-support report at $t=2$ does not become an additional communication channel that induces new beliefs or a new mechanism.

I will write $x_i(\cdot\mid m_i,m_{-i})$ to implicitly denote the extension of the allocation rule, defined for bidder $i$ after the seller observes $(m_i,m_{-i})$, to the full type space $[\underline v,\overline v]$, where $m_{-i}$ denotes the profile of messages sent by bidders other than $i$. This formulation allows for the standard envelope representation of utility across the full type space $[\underline v,\overline v]$. Additionally, define
\[
\overline{x}_i(v_i\mid m_i):=\E[x_i(v_i\mid m_i,m_{-i})],
\]
which is bidder $i$'s expected allocation rule when he sends message $m_i$, averaged over the other bidders' messages $m_{-i}$ according to $\sigma_{-i}$. Define
\[
r_i(m_i):=\inf\{v\in[\underline v,\overline v]:\overline{x}_i(v\mid m_i)>0\}
\]
to be bidder $i$'s individualized reserve price. Finally, let
\[
U_i(v_i\mid m_i):=\E[u_i(v_i\mid m_i,m_{-i})],
\]
averaged over the other bidders' messages $m_{-i}$ according to $\sigma_{-i}$.
 
\subsection{Motivating example}\label{example}

Before continuing with the formal analysis, the following example serves to build some intuition for the model.

Suppose bidder valuations are distributed $v_i \sim \mathrm{Unif}[1/2, 2]$, and the seller's outside option is valued at $c=1$. The virtual type function for all bidders under the prior is $\varphi(v) = 2v - 2$. Under full commitment, the seller sets a reserve $r^*$ such that $\varphi(r^*) = 1$, so $r^* = \frac 32$, and runs the auction as long as one bidder reports being above $r^*.$

Consider a candidate symmetric equilibrium where bidders use a binary strategy with threshold $1/2<\tau<2$:
\begin{equation}\label{ex_strat}
\sigma_i(v) = \begin{cases} 
H & \text{if } v \in [\tau, 2] \\
L  & \text{if } v \in [1/2, \tau),
\end{cases}
\end{equation}
so each bidder has two possible messages $H>L.$ Suppose the seller observes the message profile $(H, L)$. Now the virtual type functions with respect to the seller's posteriors for the two bidders are:
\begin{align*}
    \varphi_H(v_1) = 2v_1 - 2,\quad  \varphi_L(v_2)=2v_2-\tau,
\end{align*}
for the bidders sending $H$ and $L$, respectively. If the seller finds it profitable to run a mechanism after $(H,L)$ (so expected revenue from the auction is at least $c=1$), then a bidder of type $\tau$ cannot be indifferent between $H$ and $L$. In particular, if he reports $H$, he will always get 0 utility, as he will never get charged less than $\tau$ for the good in an optimal mechanism. But if he reports $L$, his virtual value will be higher than the other bidder whenever his opponent has value in $[\tau,\frac{\tau+2}{2}),$ and will thus be allocated the good in these states. Therefore a bidder of type $\tau$ must get strictly positive utility under $\sigma_i$ if the seller runs an auction after $(H,L).$ For $\sigma_i$ to be part of a symmetric equilibrium, a bidder with value $\tau$ must be indifferent between $H$ and $L$. This implies that $\tau$ must be such that the seller's revenue from running a mechanism at $(H,L)$ is less than $c=1,$ a requirement that immediately pushes $\tau$ below $1.$ On the other hand, the threshold must be high enough that the seller still finds it profitable to run an auction under the $(H,H)$ profile. Straightforward computation shows $\sigma_i$ is sustainable by thresholds roughly in the range $\frac 12< \tau\leq 0.854$, and the reserve chosen on path is $r=1.$ The seller's preferred equilibrium is the highest sustainable threshold, $\tau\approx 0.854$, as this minimizes the chance of auction failure and allows her to take her outside option in states where she wouldn't be able to under a lower threshold.

Similar reasoning shows that for any candidate partitional equilibrium, where bidders sort into bins of the form $[1/2, a_1), [a_1,a_2),\dots, [a_n,2]$, the seller again cannot run the auction if she receives a message that one bidder is in $[a_n,2]$ and another is in $[a_{n-1},a_n),$ implying that the seller cannot run the auction unless both bidders report being in the top bin, $[a_n,2].$ Therefore the equilibrium is outcome-equivalent to one of the threshold equilibria described above. Of course, it is not obvious that only finitely many messages are used on-path in equilibrium, which is what Theorem \ref{myerson_characterization} shows is without loss of outcomes.

In contrast, the binary equilibria described by (\ref{ex_strat}) are significantly easier to sustain when the seller is restricted to using a single reserve price, thereby allowing higher thresholds $\tau$ to arise in equilibrium. In particular, in this example, the seller can obtain her full-commitment payoff with a threshold $\tau=\frac 32.$ Since $\tau=\frac 32>1=c,$ the seller will now want to run the auction if either bidder reports $H$, and she will choose a common reserve of $\frac32.$ If both bidders report $L$, her revenue-maximizing auction has reserve $r=\frac 34$ and generates revenue $\frac{27}{32}<1,$ so she will not run the auction under $(L,L),$ and equilibrium constraints are satisfied.

\subsection{Solution concept}\label{subsec:solution}
The baseline solution concept for this paper will be an adaptation of perfect extended Bayesian equilibrium (PEBE), from \cite{FT_PBE}, to this continuous-type model.

\begin{definition}
A pair of strategy profiles $(\sigma,\pi)$ for the bidders and seller, and a belief system $(\mu_i)_{i=1}^n$ of the seller regarding bidder valuations after receiving messages at $t=1$, constitute a \textnormal{PEBE} iff:

\begin{enumerate}
    \item \textbf{Seller optimality:} For any message profile $m=(m_1,\dots,m_n)$, the seller's strategy $\pi(m)$ maximizes her expected utility given that she can either (i) take an outside option valued at $c$, or (ii) select a DSIC mechanism $\Gamma$ that maximizes expected revenue with respect to the beliefs $(\mu_i(\cdot\mid m))_{i=1}^n$.
    
    \item \textbf{Bidder optimality:} For every $i$ and every type $v_i\in\Theta_i$, the message $\sigma_i(v_i)$ maximizes expected utility given the seller's strategy $\pi$ and the other bidders' strategies $\sigma_{-i}$:
    \[
    \sigma_i(v_i)\in\arg\max_{m_i\in M_i}
    \mathbb{E}_{v_{-i}}\left[
    u_i(v_i\mid m_i,\sigma_{-i}(v_{-i}))
    \right].
    \]
    
    \item \textbf{No signaling what you don't know:} The seller's belief about bidder $i$ depends only on $m_i$, i.e. for any two message profiles $(m_i,m_{-i})$ and $(m_i,m_{-i}')$, including off-path messages,
    \[
    \mu_i(v_i\mid m_i,m_{-i})
    =
    \mu_i(v_i\mid m_i,m_{-i}')
    =:\mu_i(v_i\mid m_i).
    \]
    Write $F_{m_i}:=\mu_i(v_i\mid m_i)$ to denote the seller's belief about bidder $i$'s valuation after he sends message $m_i\in M_i$.

    \item \textbf{Bayes' updating on path:} For every $m_i$ such that $\sigma_i(v_i)(m_i)=1$ for some $v_i\in[\underline v,\overline v]$, the belief $F_{m_i}$ is derived from the prior $F$ using Bayes' rule.
\end{enumerate}
\end{definition}

Note that the seller's optimal mechanism is always deterministic, so her strategy is always pure in a PEBE. Given the DSIC implementation of any mechanism chosen by the seller, I leave unspecified bidder beliefs about each other and strategies at $t=2$. On path, each bidder is assumed to report truthfully at $t=2$, off path bidders report according to the above discussion following assumption \ref{ass:posterior_support}. For any PEBE $(\sigma,\pi)$ with belief system $(\mu_i)_{i=1}^n$, let $A_i$ denote the on-path messages of bidder $i$ such that for all $a\in A_i$, some type $v_i$ such that $\sigma_i(v_i)(a)>0$ gets strictly positive utility, and let $N_i$ denote all other on-path messages. Though the ``no-signaling-what-you-don't-know" condition imposes a natural restriction on seller beliefs after off-path messages, prior papers on cheap talk with multiple senders have emphasized that these models raise the concern that many of its equilibria are sustained by unreasonable off-path beliefs \citep{battaglini_02, takahashi_multi_08}. The remainder of this subsection explains a refinement that makes transparent what the seller is allowed to believe after off-path messages, and all equilibria analyzed in this paper are sustainable by its restriction.

For each bidder $i$, let $u_i(v_i\mid G,m_{-i})$ denote bidder $i$'s utility if he has value $v_i$, the seller believes his type is distributed according to $G$, and the other bidders send message vector $m_{-i}$. Let
\[
U_i(v_i\mid G):=\E[u_i(v_i\mid G,\sigma_{-i}(v_{-i}))]
\]
denote the corresponding expected utility, averaged over the other bidders' on-path messages. Let $x_i(v_i\mid G)$ denote the expected allocation probability that type $v_i$ obtains when he optimally reports in the mechanism chosen after the seller's belief about bidder $i$ is $G$, again averaged over the other bidders' on-path messages. Finally, let $U_i(v_i\mid \sigma)$ denote bidder $i$'s equilibrium utility given $(\sigma,\pi)$.

\begin{definition}\label{def:refinement}
    Let $(\sigma,\pi)$ and belief system $(\mu_i)_{i=1}^n$ form a PEBE. The PEBE has \textnormal{reasonable seller beliefs} iff for every bidder $i$ and off-path message $m_i$, the support of the seller's belief $F_{m_i}=\mu_i(v_i\mid m_i)$ regarding bidder $i$'s valuation does not contain any type $v_i$ for which there exists a type $v_i'>v_i$ such that, for every candidate belief $G$ supported inside $[\underline v,\overline v]$,
    \[
    U_i(v_i\mid G)\geq U_i(v_i\mid \sigma)
    \quad\text{and}\quad
    x_i(v_i\mid G)>0
    \]
    implies
    \[
    U_i(v_i'\mid G)-U_i(v_i'\mid \sigma)
    >
    U_i(v_i\mid G)-U_i(v_i\mid \sigma).
    \]
\end{definition}

This refinement is in the spirit of the D1 criterion from \cite{banks_sobel_87}, modified to accommodate this payoff structure and costless messaging setting. The basic intuition is similar: after observing an off-path message, the seller should not believe that the message came from type $v_i$ if there is a higher type $v_i'>v_i$ with a strictly stronger incentive to send the message for every seller belief under which type $v_i$ would be willing to deviate and obtain the good with positive probability. The requirement that type $v_i$ obtain the good with positive probability is what gives the refinement any bite in this environment because many off-path beliefs can make a type weakly willing to deviate only by giving him zero allocation and hence zero utility.

To clarify what restrictions this refinement imposes on seller beliefs, I first state the following lemma, which describes a simple but important property of all PEBE.

\begin{lemma}\label{lem:common_reserve}
Let $\sigma_i$ denote the strategy of bidder $i$ coming from a PEBE. Then for any $a,a'\in A_i$, it must be that $r_i(a)=r_i(a').$
\end{lemma}

The logic is simple: if there were $a,a'\in A_i$ such that $r_i(a)<r_i(a'),$ then bidder $i$ would face a strictly lower individualized reserve price after sending $a$ compared to $a'$. Then types slightly above $r_i(a')$ would strictly prefer sending message $a,$ implying that an individualized reserve price of $r_i(a')$ cannot arise from optimal seller behavior given her beliefs after receiving $a'$ from bidder $i.$ The following proposition clarifies what beliefs the seller can have off-path according to the above refinement. Given a strategy $\sigma_i$ coming from a PEBE, let $r_i=r_i(a)=r_i(a')$ for all $a,a'\in A_i$.

\begin{proposition}\label{prop:beliefs_support}
 If the PEBE has reasonable seller beliefs, then for any off-path message $m_i$, the seller's belief $F_{m_i}$ is not supported by any value in the interval $[\underline{v}, r_i)$.
\end{proposition}

Essentially, if some type $v_i<r_i$ thought he could benefit from sending an off-path message, he must believe that the seller will form a posterior that allows him to get the good with positive probability. But then type $r_i$ of bidder $i$, who is currently also getting 0 utility under the equilibrium, would benefit even more from that off-path message. The analysis will make clear that every equilibrium described in the results for the baseline model is sustainable by beliefs of the seller, after observing an off-path message from bidder $i,$ given by the truncation of the prior $F$ to be supported on $[r_i,\overline v].$

The main results do not depend on the restriction to reasonable seller beliefs, meaning that the results apply to equilibria without the restriction as well. Definition \ref{def:refinement} simply clarifies that the existence of informative equilibria in this model does not depend on unrealistic seller behavior off-path.

\subsection{Restriction to monotone strategies}\label{subsec:monotone}

The analysis will restrict to \emph{pure monotone} strategies, in the sense that $\sigma_i$ is weakly increasing ($v'>v\Rightarrow \sigma_i(v')\geq \sigma_i(v)$). This assumption has been used to select equilibria in strategic communication models — including by \cite{chen_kartik_sobel_08}, \cite{kartik_lying_09}, \cite{Chen_2011}, and \cite{Sobel_19} — in the \cite{crawford_strategic_1982} payoff environment.\footnote{\cite{Gordon_monotone_24} provide a partial justification for restricting to monotone strategies.} Though this assumption is generally restrictive, it has a natural interpretation in this model. In particular, following the line of literature that attributes literal meaning to cheap talk messages, if $M_i=[\underline v,\overline v]$, one can interpret a bidder sending $v\in M_i$ as saying, ``My valuation is $v$, and that will be my bid if you run the sale."\footnote{\cite{kartik_lying_09} uses this interpretation as a basis for adding ``lying costs" to the standard cheap talk model, transforming it into a signaling game. Section \ref{extensions} considers this model with lying costs added. \cite{kartik_lying_09} also considers a larger message space where the sender has infinitely many ways of communicating, ``My type is $v$." Since the main equilibrium characterizations in this paper concern outcome-equivalence, I omit this distinction for expositional simplicity.} Monotonicity of strategies can then be interpreted as requiring that if a bidder with value $v$ wants to overclaim his type to be $\hat v$, then a bidder with value $v'>v$ should not want to claim to be lower than $\hat v,$ and vice versa if bidders are underclaiming. Appendix \ref{appendix:nonmonotone_examples} gives finite-type examples showing that, once monotonicity is dropped, symmetric equilibria can involve non-monotone mixed strategies and need not be outcome-equivalent to threshold equilibria.

For the remainder of the paper, I use the term \emph{symmetric equilibrium} to mean a pair of strategies $(\sigma,\pi)$ and a seller belief system forming a bidder-symmetric PEBE with reasonable seller beliefs, where each $\sigma_i$ is pure and monotone.

\section{Equilibrium analysis for unrestricted seller}\label{unrestricted}

When the seller is unrestricted in her choice of mechanism, the symmetric equilibria of the model can be characterized as follows.

\begin{theorem}\label{myerson_characterization}
Under Assumption \ref{ass:posterior_support}, restricting attention to pure monotone bidder-symmetric strategies and PEBE with reasonable seller beliefs in the sense of Definition \ref{def:refinement}, every equilibrium in the unrestricted regime is outcome-equivalent to one in which at most one mechanism runs on path, namely a second-price auction with reserve $r\geq p_F$. If a mechanism runs on path, the equilibrium is outcome-equivalent to a threshold equilibrium where bidders report whether their valuations are above some cutoff $\tau$, and the auction runs if and only if all $n$ bidders report being above the cutoff.
\end{theorem}

The intuition behind Theorem \ref{myerson_characterization} is the following. When the seller is unrestricted over her mechanism choice, bidders face an information-provision freeriding problem. In particular, if the seller were to run an auction when some bidders revealed favorable information and others revealed low values, the optimal auction would handicap bidders providing favorable information to the seller while boosting weaker bidders. Therefore, while each bidder would like the seller to receive enough favorable information to run a mechanism, he would rather this information come from another bidder. It follows that if a bidder is providing information to the seller that favors running a mechanism, that information must be pivotal to her decision to run a sale. 

The result shows that the second-price-with-reserve auction format that is well-studied and widely used in practice arises endogenously in equilibrium despite the seller's ability to design much more individualized mechanisms after receiving communication. An implication of this result is that the seller's payoff is bounded strictly below her full commitment — studied by \cite{myerson_optimal_1981} and \cite{bulow_simple_1989} — as the seller's choice of mechanism is strongly disciplined by her limited commitment and the bidders' communication incentives. Note that the equilibria described in Theorem \ref{myerson_characterization} are sustainable by off-path beliefs satisfying the refinement in the previous section: All bidder types $v\geq r$ are sending $a,$ so they cannot benefit by deviating to an off-path message if the seller forms a belief $F|_{v\geq r}$. Furthermore, given such a belief, no bidder types $v<r$ can ever get positive utility from an optimally designed mechanism run off-path.\footnote{For a formal statement, see Lemma \ref{lem:offpath_threshold_beliefs} and its proof in Appendix \ref{appendix:proofs}.}

\begin{corollary}\label{cor:full_commit}
The seller's equilibrium utility under any symmetric equilibrium is strictly less than her full-commitment payoff.
\end{corollary}

This follows immediately from the revelation principle: Any equilibrium of this model can be reproduced by a direct mechanism when the seller has full commitment, and Theorem \ref{myerson_characterization} implies the allocation rule resulting from any symmetric equilibrium disagrees with any optimal allocation rule coming from the solution in \cite{myerson_optimal_1981}. 

\begin{remark}
\label{rem:unrestricted_existence}
A simple necessary and sufficient condition for a non-babbling threshold equilibrium in the unrestricted regime is the existence of a cutoff $\tau\in(\underline v,\overline v)$ such that the seller finds it profitable to run the mechanism when all $n$ bidders are above $\tau$, but not when one bidder has value below $\tau$ and the other $n-1$ bidders have values above $\tau$. By monotonicity of the unrestricted seller's revenue in her posteriors (see \cite{hart_reny_15}), this implies that the seller also does not run after any other message profile containing a low message.
\end{remark}

\subsection{Proof}
I now outline the proof of Theorem \ref{myerson_characterization}. Recall that for any symmetric equilibrium, $A_i$ denotes the set of on-path messages of bidder $i$ such that, for each $a\in A_i$, some type sending $a$ gets strictly positive utility, and $N_i$ denotes the set of on-path messages for which every type sending such a message gets zero utility. Assume that the sets of types sending messages in $A_i$ and $N_i$ both have positive measure. If not, the same steps in the proof show that the equilibrium is outcome-equivalent to a babbling equilibrium, so the seller either does not run a mechanism or runs a second-price auction with reserve $p_F$. I will refer to messages in $A_i$ as ``active'' and messages in $N_i$ as ``nonactive.'' The proof proceeds by reducing the communication structure to a binary one and then showing that the seller can run only when all bidders send their active messages.
\begin{enumerate}
    \item Every active message gives bidder $i$ the same individualized reserve price $r:=r_i(a)$.
    \item Consequently, there can be only one active message $a_i$ for each bidder.
    \item A bidder who sends a nonactive message must be excluded from any mechanism the seller runs with probability 1. Consequently, the seller's behavior after receiving a nonactive message from bidder $i$ is independent of the specific nonactive message sent, and nonactive messages can be pooled into a single message while preserving outcomes.
    \item Given a pooled nonactive message $\bar n_i$, the seller cannot run the auction after any message profile in which some bidder sends $\bar n_i$. Hence the seller can run only after the all-active profile $(a_1,\dots,a_n)$.
\end{enumerate}

\begin{lemma}\label{lem:one_active}
There cannot be two distinct active messages $a,a'\in A_i$.
\end{lemma}
\begin{proof}
    Suppose there were two distinct active messages $a,a'\in A_i$ and WLOG assume $a>a'.$ Since $a$ is an active message it must be sent by some type $v>r.$ Let $\supp(F_a)$ denote the support of $F_a.$ By the same reasoning as lemma \ref{lem:common_reserve}, it cannot be that $\inf(\supp(F_a))>r,$ so $\inf(\supp(F_a))\leq r.$ Monotonicity of strategies implies that $\supp(F_a)$ is connected, so it must contain the interval $(\inf(\supp(F_a)),v].$ Since $a>a'$, this implies $\sup(\supp(F_{a'}))\leq \inf(\supp(F_a))\leq r,$ implying that $a'$ is not an active message.
\end{proof}

For the remainder of the proof, let the active message be $a,$ and label it $a_i$ to specify that bidder $i$ sends it.

\begin{lemma}\label{lem:exclusion}
If the seller runs a mechanism after a message profile in which bidder $i$ sends a nonactive message $n\in N_i$, then bidder $i$ is excluded from the mechanism with probability 1, i.e.
\[
x_i(v\mid n,m_{-i})=0
\quad\text{for }F_n\text{-a.e. }v
\] and all positive-probability opponent message profile $m_{-i}$ for which the mechanism is run.
\end{lemma}

The argument is straightforward and follows from the envelope representation of bidder utilities and IIR. A simple intuition is the following. Suppose bidder $i$ sends a nonactive message $n$ and obtains the good with positive probability in some mechanism run after $(n,m_{-i})$. Since $n\in N_i$, every type sending $n$ receives zero interim utility. Thus any type $v$ sending $n$ who is allocated the good must be charged exactly $v$ by IIR. If this occurs for some type $v<r$, then type $r$ — who also receives zero utility in equilibrium — could deviate to $n$ obtain strictly positive utility, which would cause the equilibrium to unravel.

By Lemma \ref{lem:exclusion}, since a bidder sending a nonactive message is excluded from any mechanism the seller runs with probability 1, the seller's optimal mechanism after receiving a nonactive message from bidder $i$ depends only on her posteriors about the remaining bidders. In particular, fixing the messages of the other bidders, the seller's behavior is independent of which nonactive message bidder $i$ sent. This allows all nonactive messages to be pooled into a single message $\bar n_i$ without changing any outcomes.

\begin{lemma}\label{lem:pool_nonactive}
Any symmetric equilibrium is outcome-equivalent to one in which each bidder $i$'s message space is replaced by $\{a_i,\bar n_i\}$, where all messages in $N_i$ are pooled into a single message.
\end{lemma}

The previous lemmas already establish that the communication structure when bidders use monotone strategies is \emph{binary}. It remains to show that the seller can only run a mechanism when all bidders use their active messages.

\begin{lemma}\label{lem_active_nonactive}
The seller cannot run a mechanism after any message profile in which some bidder sends his nonactive message.
\end{lemma}

This is the key step in the proof of Theorem \ref{myerson_characterization}. While the previous lemmas reduce communication to a binary active/nonactive structure, this lemma shows that the seller cannot use this information to run a mechanism unless every bidder sends the active message. In particular, the seller is not allowed to run a mechanism only among bidders using their active message while excluding bidders sending their nonactive message. As the proof will show, this is because such exclusion is itself inconsistent with Myerson-optimality.

To see the logic, suppose the seller runs after a profile in which bidder $i$ sends $\bar n_i$. By Lemma \ref{lem:exclusion}, bidder $i$ must be excluded from the mechanism with probability 1. This implies that his value must be below $r-\epsilon$ for some $\epsilon>0$, which follows from the classical mechanism design result that virtual values are ``undistorted at the top" of the support of a distribution and distorted downwards at the bottom. Since the seller is running the auction, at least some of the remaining bidders must be sending their active message, which must be sent by a positive measure of types below $r$. But then, by the definition of $r,$ the seller's mechanism must be withholding the good from all bidders on a positive-probability event where the bidders using their active message have a value below $r$ and a bidder using his nonactive message has a positive virtual type. This contradicts the optimality of the seller's mechanism.

What remains to show of Theorem \ref{myerson_characterization} is that $r\geq p_F$, which follows directly from the fact that the optimal reserve in an auction with symmetrically distributed bidders is the monopoly price.
\begin{lemma}\label{lem:r_pF}
Take any symmetric equilibrium of the form described in lemma \ref{lem:pool_nonactive}. Then $r\geq p_F.$
\end{lemma}
The conclusion of Theorem \ref{myerson_characterization} now follows immediately from the previous two lemmas and the monotonicity of bidder strategies.

\section{Restricted seller}\label{restricted}
Most of the analysis so far has focused on the case where the seller is unrestricted at $t=2$ in terms of what mechanisms she can run. This section considers the case where the seller ex-ante commits to running an auction with a \emph{common reserve}, ruling out discriminatory auction formats. This restriction effectively introduces limited, but neither full nor zero, cross-period commitment on the part of the seller, which could come from legal constraints a seller faces in practice.\footnote{Even without such explicit institutional or legal constraints, such commitments can come from other dynamic considerations. For example, one possible microfoundation is the following. Suppose the seller is long-lived and faces new cohorts of bidders over time. If deviations from an announced restriction are publicly observed, then bidders in future periods can punish a deviation by coordinating on a collusive/seller-worst equilibrium. Such a punishment sustains restrictions on the seller's behavior whenever the gain from violating the restriction is smaller than the discounted continuation loss from future punishment.}

From this point through the extensions, I specialize to the case of two bidders for ease of exposition. It turns out that the seller can do better — and sometimes attain her commitment payoff — in this case. The next lemma establishes the equivalent of Theorem \ref{myerson_characterization} for the restricted regime.

\begin{lemma}\label{lem:restricted_threshold_wlog}
In the restricted regime, every symmetric equilibrium is outcome-equivalent to one in which bidders use threshold strategies.
\end{lemma}

Accordingly, it is without loss of outcomes to focus on binary strategies with threshold $\underline v <\tau<\overline v$:
\begin{equation}\label{binary_strats}
\sigma_i(v) = \begin{cases} 
H & \text{if } v \in [\tau, \overline v] \\
L  & \text{if } v \in [\underline v, \tau),
\end{cases}
\end{equation}
which will henceforth be called symmetric threshold strategies. This section derives some comparative statics and practical qualitative insights.

Let $\Rev_\tau (H,H)$ denote the seller's revenue from a second-price auction with an optimal common reserve when bidders are using symmetric threshold strategies with threshold $\tau$, and both bidders report $H.$ Similarly define $\Rev_\tau (H,L)$ and $\Rev_\tau(L,H).$ Let $r^*=\varphi^{-1}(c)$ denote the seller's optimal reserve coming from her full-commitment solution: That is, she commits to running the auction with reserve $r^*$ if and only if at least one bidder reports being above $r^*.$ The following result characterizes the symmetric threshold equilibria under the seller's restricted regime.

\begin{proposition}\label{prop:restricted_char}
    In any symmetric threshold equilibrium where the seller commits to using a common reserve at $t=2$, the following hold when $n=2$.
    \begin{enumerate}
        \item If an auction runs on-path, either: (i) an auction with reserve $r$ is run iff at least one bidder reports $H$, or (ii) an auction with reserve $r$ is run iff both bidders report $H.$ That is, at most one mechanism runs on path.
        \item If $\tau\leq p_F$, the reserve $r$ chosen must be $p_F$ if an auction runs on-path. If $\tau>p_F$, the reserve chosen in any auction run is $\tau$.
    \end{enumerate}
\end{proposition}

In other words, the restriction to a common reserve price allows the auction to run much more often in equilibrium. In the unrestricted regime, the seller's ability to tailor mechanisms creates an incentive to exploit the specific information revealed by the message profile, like setting a higher reserve for a bidder revealing a high type. As shown in Theorem \ref{myerson_characterization}, this temptation limits the equilibrium to ones where the auction runs only when both bidders report a high valuation.

However, when the seller is restricted to a second-price auction with a common reserve, she loses the ability to exploit a high-type bidder with an individualized high reserve while offering a lower reserve to a low-type bidder. This strengthening of her limited-commitment power softens IC constraints for the bidders. High-value bidders are more willing to disclose their type because they know the reserve price they face cannot exceed the common reserve, which must also apply to the other bidder. Consequently, equilibria where the auction runs even if only one bidder reports $H$ become sustainable, significantly expanding the seller's attainable payoffs.

I now solve for optimal thresholds in both regimes and compare them.

Let $\Pi_1(\tau)$ denote the seller's payoff if she runs the auction iff at least one bidder reports $H$ when the threshold is $\tau$, and takes her outside option otherwise --- ignoring for the moment whether this behavior is sequentially rational. Similarly, let $\Pi_2(\tau)$ denote the seller's payoff if she runs the auction iff both bidders report $H.$ Let $G_1(v)=F(v)^2$ be the CDF of $v_{(1)}=\max\{v_1,v_2\}$ and $G_2(v)=1-(1-F(v))^2$ be the CDF of $v_{(2)}=\min\{v_1,v_2\}.$ Then
\[
\Pi_1(\tau)=\int_{\tau}^{\overline v}\max\{\varphi(x),0\}\,g_1(x)\,dx+cG_1(\tau),
\]
and
\[
\Pi_2(\tau)=\int_\tau^{\overline v}\mathbb{E}[\max\{\varphi(v_{(1)}),0\}\mid v_{(2)}=x]\,g_2(x)\,dx+cG_2(\tau).
\]
Let $r^*=\varphi^{-1}(c)$ denote the seller's optimal reserve coming from her full-commitment solution: That is, she commits to running the auction with reserve $r^*$ if and only if at least one bidder reports being above $r^*.$ Finally, let $\Rev_\tau^u(H,L)$ denote the seller's revenue from her optimal unrestricted mechanism after the message profile $(H,L)$ under threshold $\tau$.

The next proposition describes the seller-optimal threshold in the unrestricted regime.

\begin{proposition}\label{prop:unrestricted_optimal}
Under the unrestricted regime, the seller-optimal threshold equilibrium is supported by the highest threshold $\tau_U^*$ for which the seller does not want to run the auction under an $(H,L)$ message profile. Equivalently,
\[
\tau_U^*=\sup\{\tau:\Rev_\tau^u(H,L)\le c\}.
\]
\end{proposition}

\begin{proof}
By Theorem \ref{myerson_characterization}, every non-babbling equilibrium in the unrestricted regime is outcome-equivalent to a threshold equilibrium where the auction runs iff both bidders report $H$. Therefore the seller's payoff from any such equilibrium is $\Pi_2(\tau)$. Differentiating gives
\[
\Pi_2'(\tau)=g_2(\tau)\left(c-\mathbb{E}[\max\{\varphi(v_{(1)}),0\}\mid v_{(2)}=\tau]\right).
\]
Hence $\Pi_2$ is single-peaked. To locate the sustainable thresholds relative to the peak, observe that
\[
\mathbb{E}[\max\{\varphi(v_{(1)}),0\}\mid v_{(2)}=\tau]=\Rev_\tau(H,L).
\]
Since the unrestricted seller always has access to the best common-reserve auction, it must be that
\[
\Rev_\tau^u(H,L)\geq \Rev_\tau(H,L).
\]
Thus if $\tau$ is a sustainable threshold under the unrestricted regime, \[\Rev_\tau(H,L)\leq c.\] It follows that the sustainable unrestricted thresholds lie on the increasing part of $\Pi_2$, so the seller prefers the highest sustainable threshold.
\end{proof}

The next proposition then solves for the seller-optimal threshold under the restricted regime.
\begin{proposition}\label{prop:seller_optimal}
Under the restricted regime, seller's optimal threshold $\tau_S^*$ within the class of symmetric threshold equilibria is $\tau_S^*=r^*$ if $\Rev_{r^*}(L,L)\leq c.$ Otherwise, it is the highest threshold sustainable where an auction runs iff at least one bidder reports $H$. That is, $\Rev_{\tau_S^*}(L,L)=c.$
\end{proposition}
In particular, the seller never prefers a lower threshold $\tau$ where both bidders would need to report $H$ in order to induce a mechanism even though such a threshold could result in profitable trades that would not occur under a higher threshold.

\begin{proof}
       Using the fact that the virtual type of a bidder with valuation $v\geq \tau$ with respect to the truncation $F\mid_{v\geq \tau}$ is exactly $\varphi(v)$, taking derivatives yields, \begin{equation}\label{eqn:pi_1}
    \Pi_1'(\tau)=g_1(\tau)\left(c-\max\{\varphi(\tau),0\}\right),
    \end{equation} and \begin{equation}\label{eqn:pi_2}
        \Pi_2'(\tau)=g_2(\tau)\left(c-\mathbb{E}[\max\{\varphi(v_{(1)}),0\}\mid v_{(2)}=\tau]\right).
    \end{equation}
    Therefore $\Pi_1$ and $\Pi_2$ are single-peaked and maximized at $\tau_1^*=r^*$ and $\tau_2^*$ given by \begin{equation}\label{eq:tau_2constraint}
        c=\mathbb{E}[\max\{\varphi(v_{(1)}),0\}\mid v_{(2)}=\tau_2^*], 
    \end{equation}
    respectively.
    
    For threshold $\tau$ and seller behavior consistent with $\Pi_1$ to be part of an equilibrium, it must be sequentially rational for the seller to not run the auction when she receives the message profile $(L, L)$. This imposes the constraint:
\begin{equation}\label{seller_IC_restricted}
        \Rev_\tau(L, L) \le c.
\end{equation}

Thus if $\Rev_{r^*}(L, L) \le c$, the threshold $\tau_1 = r^*$ is sustainable in equilibrium, and is optimal among equilibria where the seller runs an auction iff at least one bidder reports $H$ by (\ref{eqn:pi_1}). Since $\tau_1 = r^*$ achieves the full-commitment payoff, it is necessarily the optimal equilibrium threshold $\tau_S^* = r^*$ among all symmetric equilibria.

On the other hand, if $\Rev_{r^*}(L, L) > c$, the threshold $\tau = r^*$ is not sustainable because the seller would prefer to deviate and run the auction even when both bidders report $L$. To restore equilibrium, $\tau$ must be lowered to satisfy the constraint, since $\Rev_\tau(L, L)$ is strictly increasing in $\tau$. Since $\Pi_1(\tau)$ is single-peaked in $\tau$, the seller's preferred threshold among equilibria where she runs an auction whenever one bidder reports $H$ is the highest possible $\tau$ that satisfies (\ref{seller_IC_restricted}), i.e. $\tau_S^*$ satisfying $\Rev_{\tau_S^*}(L, L) = c$. To see that such a $\tau_S^*$ is actually optimal, it remains to check that it does better than any equilibrium where the seller runs an auction iff both bidders report $H.$ By above, the optimal threshold under this set of equilibria is $\tau_2^*$ given by (\ref{eq:tau_2constraint}). Now observe that \[c=\mathbb{E}[\max\{\varphi(v_{(1)}),0\}\mid v_{(2)}=\tau_2^*]=\Rev_{\tau_2^*}(H,L),\] so the threshold $\tau_2^*$ also sustains an equilibrium where the seller runs the auction iff at least one bidder reports $H.$ Therefore $\tau_S^*$ is the optimal equilibrium threshold for the seller.
\end{proof}

The three cutoff values $\tau_U^*$, $\tau_2^*$, and $\tau_S^*$ give a simple comparison of the two regimes.

\begin{proposition}\label{prop:compare_thresholds}
Every threshold equilibrium under the unrestricted regime is also a threshold equilibrium under the restricted regime. Moreover,
\[
\tau_U^*<\tau_2^*\le \tau_S^*.
\]
Thus the restricted regime strictly expands the set of sustainable thresholds, and the seller benefits from the restriction. In particular, every threshold sustainable under the unrestricted regime is also sustainable under the restricted regime. Moreover, the restricted regime admits strictly higher thresholds: first within the class of equilibria where the auction runs iff both bidders report $H$, and then beyond that within the class where the auction runs iff at least one bidder reports $H$. Thus there is a strong sense in which the restriction to a common reserve improves the seller's payoff.
\end{proposition}

While Corollary \ref{cor:full_commit} established that an unrestricted seller \emph{always} earns strictly less than her full-commitment payoff, Proposition \ref{prop:seller_optimal} shows that committing to a common reserve format allows her to achieve her full-commitment payoff as long as the outside option $c$ is sufficiently valuable to dissuade the seller from running the auction when both bidders report $L$. However, the seller's limited commitment still binds in some cases: if $c$ is not valuable enough to sustain the full-commitment solution as an equilibrium, the threshold $\tau$ must be lowered below $r^*$ to deter deviations by the seller. It is easy to see that as the value of the outside option decreases, the seller's optimal threshold $\tau_S^*$ also decreases. Furthermore, the following holds.
\begin{proposition}
\label{prop:comparative_statics}
Consider the seller-optimal symmetric threshold equilibrium $\tau_S^*$ in the restricted regime. As the value of the outside option $c$ decreases:
\begin{enumerate}
    \item The equilibrium threshold $\tau_S^*$ strictly decreases, and the seller's payoff decreases.
    \item The ex-ante probability that the seller runs an auction strictly increases.
    \item The reserve price $r(\tau_S^*)=\max\{\tau_S^*,p_F\}$ strictly decreases as long as $\tau_S^*>p_F$, and remains constant at $p_F$ once $\tau_S^*\leq p_F$.
    \item Conditional on the seller running an auction, the probability that the auction results in no trade is constant while $\tau_S^*\geq p_F$, and strictly increases once $\tau_S^*<p_F$.
\end{enumerate}
\end{proposition}

In particular, as $c$ falls, $\tau_S^*$ falls below $p_F$, increasing the mass of participating bidder types below the reserve.

The optimal equilibrium thresholds can be used to formalize the value of the communication stage to the seller. Without a communication stage, the seller cannot condition her decision on bidder messages, so she either runs the optimal auction under the prior or takes the outside option. Let
\[
R_0:=\int_{\underline v}^{\overline v}\max\{\varphi(x),0\}\,g_1(x)\,dx
\]
denote the seller's expected revenue from running the optimal auction under the prior. Her payoff without a communication stage is therefore
\[
\Pi^N(c)=\max\{R_0,c\}.
\]
Let $R_L(\tau):=\Rev_\tau(L,L)$, and write the seller-optimal threshold in the restricted regime as a function of $c$:
\[
\tau_S(c)=
\begin{cases}
\varphi^{-1}(c) & \text{if } R_L(\varphi^{-1}(c))\leq c,\\
R_L^{-1}(c) & \text{if } R_L(\varphi^{-1}(c))>c.
\end{cases}
\]
The seller's payoff with communication is
\[
\Pi^C(c)=\Pi_1(\tau_S(c);c),
\]
where
\[
\Pi_1(\tau;c)=\int_{\tau}^{\overline v}\max\{\varphi(x),0\}\,g_1(x)\,dx+cG_1(\tau).
\]
Thus the value of the communication stage to the seller is
\[
\Delta(c):=\Pi^C(c)-\Pi^N(c).
\]

The analysis earlier in the section immediately yields the following proposition.

\begin{proposition}\label{prop:value_communication}
The value of the communication stage is given by
\[
\Delta(c)=
\begin{cases}
\Pi_1(R_L^{-1}(c);c)-R_0 
& \text{if } R_L(\varphi^{-1}(c))>c,\\
\Pi_1(\varphi^{-1}(c);c)-R_0 
& \text{if } R_L(\varphi^{-1}(c))\leq c\leq R_0,\\
\Pi_1(\varphi^{-1}(c);c)-c 
& \text{if } c\geq R_0.
\end{cases}
\]
Moreover, $\Delta(c)$ is strictly increasing for $c<R_0$ and strictly decreasing for $c>R_0$. Hence the communication stage is most valuable when
\[
c=R_0,
\]
that is, when the seller is exactly indifferent under the prior between running the auction and taking the outside option.
\end{proposition}

The previous two propositions help connect the model to documented empirical patterns in mergers and acquisitions. In particular, \cite{MS_18} examine the differences between \emph{target-initiated} and \emph{bidder-initiated} deals. Target-initiated deals begin with the firm approaching potential bidders, while bidder-initiated deals begin with an acquirer approaching the firm with an indication of interest or initial offer. \cite{MS_18} find that target-initiated deals are more common among firms with lower profitability and poor recent stock returns, which in the model corresponds to a lower outside option $c$. Interpreted through the lens of the model, the target approaches potential bidders because her outside option is not very attractive. In the notation above, this corresponds to the region $c<R_0$, where the seller would run a sale process even under the prior. In this region, the communication stage does not primarily matter because it changes whether the seller sells at all. Rather, it helps the seller because it lets her condition the threshold and reserve on pre-sale information.

On the other hand, bidder-initiated deals correspond to cases where the seller's outside option value is high, which can be interpreted as maintaining the ``status quo" option. In this parameter region, $c>R_0$, the seller would not run an auction based on the prior. A bidder's initial approach can then be interpreted as a favorable communication event that induces the seller to run a sale process she otherwise would not have started. Since, as $c$ falls, Proposition \ref{prop:comparative_statics} says that the equilibrium threshold and reserve fall with $c$, the model gives a clear mechanism through which target-initiated deals can have lower final deal prices compared to bidder-initiated deals, as documented by \cite{MS_18}. Once $\tau_S(c)<p_F$, the reserve remains fixed at $p_F$ while the distribution of participating bidders shifts downward, so conditional on initiating a sale process, the probability that the auction fails increases, consistent with the lower completion rates in target-initiated processes documented by \cite{liu_officer_2019}.

\section{Extensions and robustness}\label{extensions}
Thus far, this paper has focused on the case where bidder messaging is completely costless, and bidders can decide whether or not to participate in a seller's sequentially-optimal mechanism after it is chosen. However, the setup and equilibria analyzed raise several robustness concerns. First, why would a bidder choose to participate in this communication stage at all if sending the messages were costly? Indeed, if there was any positive cost $\epsilon$ to sending the communication messages at $t=1$, the equilibria described in the previous sections would unravel if bidders are given the decision to participate or not in the communication stage and there are no other changes to the model. This is because a seller maximizing revenue after receiving messages must leave a positive measure of bidder types sending those messages interim utility below $\epsilon$, and therefore these bidder types would not participate. Additionally, the threshold equilibria described are sustained by complete indifference between the types below the threshold regarding whether they should report truthfully or not. Motivated by these concerns, this section shows that some of the equilibria described in the previous sections arise as the limit of equilibria of some natural perturbations of the model, but that some thresholds are more fragile than others, and that the robustness of a threshold equilibrium depends on which seller regime is being analyzed. As in the previous section, I focus on the $n=2$ case.

\subsection{Entry costs}
The first perturbation of the model treats first-period communication as costly entry into the sale process.  Sending any message in $M_i$ requires the bidder to post a refundable deposit of $\epsilon>0$ to the seller, while sending $\varnothing$ is costless. If a bidder sends $\varnothing$, he cannot participate in any mechanism the seller chooses. If the seller takes the outside option, all deposits are refunded. This extension captures a small entry-decision friction, with the idea that in some applications a bidder who indicates serious interest may need to enter a formal process at some small cost.

More concretely, after observing $(m_1,m_2)$, the seller can either take the outside option and refund the deposits or choose a mechanism within the unrestricted or restricted class described previously, subject to the following modified IIR constraint. If bidder $i$ made a deposit by sending a message $m_i\in M_i$, and the seller runs a mechanism after $(m_i,m_{-i})$, then
\[
    u_i(v_i\mid m_i,m_{-i})\geq \epsilon
    \qquad
    \text{for all }v_i\in\supp(F_{m_i}).
\]
The modified IIR constraint captures the idea that the seller must leave enough surplus to even the lowest participating bidder types to induce them to post the deposit. In some profiles, this required surplus can be delivered by lowering the reserve price slightly. In other profiles, it requires a literal cash subsidy to bidders at the mechanism stage.\footnote{As with the common-reserve restriction, the required limited commitment on the part of the seller can be given a repeated-game interpretation. If the seller is long-lived and deviations from promised entry terms are publicly observed, future bidder pools can punish such deviations by coordinating on seller-worst equilibria. The one-period gain from violating the modified IIR floor is of order $\epsilon$, so the commitment required vanishes as $\epsilon\to0$. Note also that in the restricted regime, the common-reserve restriction applies to the allocation rule, while transfers may include fixed rebates or purchase-price credits used to satisfy the modified IIR constraint.} Let $\Gamma_\epsilon$ denote the perturbed game.

Note that the perturbation does not change the seller's continuation decisions, but rather bidders' incentives to send costly messages.

\begin{proposition}\label{prop:entry_costs}
Let $\sigma$ be a nondegenerate symmetric threshold equilibrium of the baseline game with cutoff $\tau$.

If $\sigma$ is a restricted-regime equilibrium in which an auction runs iff at least one bidder reports being weakly above $\tau$, then for all sufficiently small $\epsilon>0$, there is an equilibrium $\sigma_\epsilon$ of $\Gamma_\epsilon$ such that $\lim_{\epsilon\to0}\sigma_\epsilon=\sigma$, after identifying $\varnothing$ with the low message. If $\tau\geq p_F$, the mechanisms can be chosen so that every type below $\tau$ strictly prefers $\varnothing$. If $\tau<p_F$, these low types are only weakly willing to use $\varnothing$.

If $\sigma$ is an unrestricted-regime equilibrium, or a restricted-regime equilibrium in which an auction runs iff both bidders report being weakly above $\tau$, then for all sufficiently small $\epsilon>0$, there is an equilibrium $\sigma_\epsilon$ of $\Gamma_\epsilon$ such that $\lim_{\epsilon\to0}\sigma_\epsilon=\sigma$, after identifying $\varnothing$ with the low message. In these equilibria, low types are only weakly willing to use $\varnothing$.
\end{proposition}

The proposition distinguishes which threshold equilibria are robust to small entry costs. When one high message triggers the auction and $\tau\geq p_F$, a low type who deviates upward sometimes becomes the sole entrant. In that state, the seller can satisfy modified IIR for the cutoff type by lowering the mechanism-stage reserve from $\tau$ to $\tau-\epsilon$ and retaining the deposit. The total price faced by the deviator is still $\tau$, so every type below $\tau$ is strictly worse off.

When $\tau<p_F$, the reserve remains above the cutoff, so lowering the reserve slightly does not give the cutoff type any surplus. The seller must instead use a rebate or fixed credit, which a low type can also collect after deviating upward. This is also the case when two high messages are needed to trigger the auction, since the cutoff type always wins the good with zero probability in such equilibria. Thus, a deviating low type either has his deposit refunded because no auction is run, or receives the same rebate as the entrants when the auction is run. These
equilibria can still be approximated as $\epsilon\to0$, but the indifference of low types remains.

\subsection{Lying costs}
In some applications, endowing the seller with commitment power to select a mechanism that leaves all types with some surplus may be unrealistic. A natural alternative to break the indifference of low types, without giving the seller additional commitment power, is to introduce a small cost of lying \citep{kartik_lying_09}. Lying costs seem natural in many applications: For example, in high-stakes settings, a bidder who overclaims his type too much may worry about reputational damage from submitting a much lower bid in the resulting auction.

Suppose the message space is simply the type space, $M_i = [\underline{v}, \overline{v}]$. A bidder of true type $v$ who sends message $m$ incurs a direct disutility $\epsilon \ell(m, v)$, where $\epsilon > 0$ scales the cost of lying. Assume the lying cost function $\ell(m,v)$ is twice continuously differentiable, strictly submodular and satisfies $\ell(v,v) = 0$ and $\ell(m,v) > 0$ for $m \neq v$.\footnote{These are the same assumptions on $\ell$ used in \cite{kartik_lying_09}.}

The form of the equilibria of the perturbed game matches the equilibria \cite{kartik_lying_09} focuses on: types below $\tau$ disclose their type, and types $v\geq \tau$ pool at $\tau.$ If a bidder sends an off-path message $m>\tau$, the seller forms a point-mass belief on $\overline v$ about his type. Since types below $\tau$ now truthfully reveal their valuations, the seller's optimal mechanism after observing one bidder's report $v<\tau$ and another bidder's pooled message $\tau$ depends on the regime. Under the restricted regime, the seller must use a common reserve, so her mechanism takes the same form regardless of the specific low report. In the unrestricted regime, however, the seller can design a Myerson-optimal mechanism exploiting her knowledge of one bidder's exact type. Accordingly, the two regimes require separate treatment.

\begin{proposition}\label{prop:lying_costs_restricted}
Let $\sigma$ be a symmetric threshold equilibrium under the restricted regime in the baseline model with cutoff $c\geq \tau \geq p_F$. For any lying cost function $\ell$ satisfying the properties above, there exists $\bar{\epsilon} > 0$ such that for all $\epsilon < \bar{\epsilon}$, there is an equilibrium $\sigma_\epsilon$ of the perturbed game $\Gamma_{\ell, \epsilon}$ such that $\lim_{\epsilon \to 0} \sigma_\epsilon$ is outcome-equivalent to $\sigma$. 
\end{proposition}

Under the restricted regime, the seller must use a common reserve, so her optimal mechanism after observing one bidder's truthful report $v<\tau$ and another bidder's pooled message $\tau$ is a second-price auction with reserve $\tau$ — the same mechanism as in the baseline threshold equilibrium. Since $c\geq \tau$, the seller does not find it profitable to run the auction when both reports are below $\tau$, so her behavior matches the baseline. Verifying bidder incentives is straightforward: types $v<\tau$ strictly prefer truthful reporting because deviating to $\tau$ yields zero mechanism utility but incurs a positive lying cost, and for types $v>\tau,$ because the lying cost function is flat near the truth, a simple Taylor expansion can be used to bound lying costs below the mechanism utility when $\epsilon$ is sufficiently small.

\begin{proposition}\label{prop:lying_costs_unrestricted}
Let $\sigma$ be a symmetric threshold equilibrium under the unrestricted regime in the baseline model with cutoff $c\geq \tau \geq p_F$. Suppose additionally that \begin{equation}\label{eq:lying_unrestricted_condition}
c\geq \E[\max(\tau,\varphi(v))\mid v\geq \tau],
\end{equation} where $v\sim F.$ For any lying cost function $\ell$ satisfying the properties above, there exists $\bar{\epsilon} > 0$ such that for all $\epsilon < \bar{\epsilon}$, there is an equilibrium $\sigma_\epsilon$ of the perturbed game $\Gamma_{\ell, \epsilon}$ such that $\lim_{\epsilon \to 0} \sigma_\epsilon$ is outcome-equivalent to $\sigma$.
\end{proposition}

Condition (\ref{eq:lying_unrestricted_condition}) ensures that the seller prefers her outside option whenever one bidder's type is known to be some $v<\tau$ while the other bidder's type is weakly above $\tau$.

Note that neither statement can hold for thresholds $\tau<p_F$, because a positive measure of types above the threshold would receive zero interim utility from the mechanism and therefore cannot agree on a single pooled message in the presence of lying costs.

\subsection{Discussion}
The previous two extensions provide methodologies to further refine the equilibrium set of the baseline model, depending on which factors one considers realistic in applications. A conclusion of both extensions is that when the seller's outside option is weaker (lower $c$), the equilibria of the baseline model are more fragile. In the case of the entry-cost extension, whenever $c<\Rev_{p_F}(L,L)$, all symmetric equilibria of the baseline model when the seller is restricted involve thresholds $\tau<p_F$. By Proposition \ref{prop:entry_costs}, this means all types below the threshold are still indifferent between entering the perturbed game with entry costs or not. The fragility of the threshold equilibria when $c$ is low is also apparent when considering the model with lying costs, where the threshold equilibria examined in the baseline model fail to hold up at all whenever $c<p_F$.

The value of $c$ can be interpreted as a measure of the incentive \emph{alignment} between the seller and bidders necessary to make cheap talk informative, since a higher $c$ gives bidders a greater incentive to reveal their values to induce an auction process. The extensions analyzed therefore suggest that equilibria of the baseline model when the seller and bidders are less aligned are inherently more fragile.

However, while a weaker outside option necessitates lower equilibrium thresholds, it does not conversely imply that higher thresholds are always more robust. Indeed, equilibrium existence in the lying cost extension requires the threshold to be bounded above by the value of the outside option itself ($\tau \le c$). If an equilibrium threshold is pushed too high, the seller-bidder alignment breaks down in the opposite direction: When there are lying costs, low types truthfully report their valuations, so the seller will inevitably observe intermediate types $v \in (c, \tau)$. Because these valuations exceed her outside option, and the seller cannot commit to walking away, she will deviate to run an extractive mechanism. This tradeoff implies that when there are lying costs, the seller will \emph{never} be able to obtain her commitment payoff in the restricted regime, as illustrated in section \ref{restricted}, since doing so necessarily entails a threshold $\tau=r^*>c$. In fact, in the motivating example found in section \ref{example}, only one threshold equilibrium in the restricted regime can be sustained if there are lying costs — namely, $\tau=c=1$ — and no thresholds can be sustained if the seller is unrestricted.

\section{Conclusion}

This paper provides a theoretical framework for analyzing pre-auction communication as an extensive-form game between bidders sending cheap talk messages and a seller with limited commitment. The analysis demonstrates that despite the potential for rich information transmission, the equilibrium communication is coarse. When the seller is unrestricted in her choice of mechanism, the bidders' strategic tradeoff between inducing the seller to hold an auction and avoiding the extraction of information rents at the mechanism stage forces the equilibrium to take a simple binary structure. A comparison of these equilibria to the case where the seller can commit to using a single reserve price provides a novel justification for the prevalence of simple auction formats — specifically, second-price auctions with reserves — even when more complex mechanisms are theoretically optimal.

The model has some obvious limitations in the context of aforementioned applications. For example, in many practical settings, bidder valuations likely contain a common value component as well, unlike the independent private values environment adopted by this paper. The presence of common values would introduce additional strategic considerations: For instance, a bidder's natural incentive to underreport his signal in common values auctions might be exacerbated if revealing high information also invites more aggressive competition from the seller or other bidders who infer the high common value. The seller's ability to gain information about the common value component from bidder communications, however, can potentially strengthen the threshold structure of the communication equilibria. In this paper, the private values framework serves as a useful benchmark — isolating the friction of limited commitment from other informational externalities — while also capturing the high potential for idiosyncratic preferences for bidders in many markets where these pre-sale communication practices are widespread.
\appendix

\section{Appendix: Omitted proofs}
\label{appendix:proofs}
\begin{proof}[Proof of lemma \ref{lem:common_reserve}.]

Suppose for contradiction that there exist $a,a'\in A_i$ with $r_i(a)=r<r'=r_i(a')$. Clearly $U_i(r'\mid a')=0$. Consider $U_i(r'\mid a)$. I first claim that $U_i(r'\mid a)>0.$ Indeed, if $U_i(r'\mid a)=0$ (by interim individual rationality it cannot be negative), then $u_i(r'\mid a,m_{-i})=0$ for every opponent message profile $m_{-i}$ that occurs on-path. Fix any such $m_{-i}$ for which the seller runs the auction after $(a,m_{-i})$. Such an $m_{-i}$ must exist, since $a\in A_i$ implies there is some type who obtains positive utility after sending $a$ and this requires that the auction is run with positive probability after $a$. Then $
u_i(r'\mid a,m_{-i})=0$ 
implies $x_i(t\mid a,m_{-i})=0$ for all $t\in[r,r']$. Since this is true for all $m_{-i}$ for which the seller runs the auction after $(a,m_{-i})$, this contradicts the definition of $r_i(a).$

Thus $U_i(r'\mid a)>0,$ so type $v=r'$ strictly prefers sending $a$ over $a'$. Since $U_i(\cdot\mid a)$ is continuous in $v$, there exists $\varepsilon>0$ such that every $v\in[r',r'+\varepsilon]$ strictly prefers $a$ over $a'$, implying $\sigma_i(v)(a')=0$ for all $v\in[r',r'+\varepsilon]$. 
I claim that this contradicts $r_i(a')=r'$. Fix any opponent message profile $m_{-i}$ that occurs on-path such that the seller runs the auction after $(a',m_{-i})$. Since the seller chooses a Myerson-optimal auction, bidder $i$ reporting $a'$ faces an interim allocation rule $x_i(\cdot \mid a',m_{-i})$ such that there is a threshold $$\kappa=\inf\{v_i\in \mathrm{supp}(F_{a'})\mid x_i(v_i \mid a',m_{-i})>0\},$$ and bidder $i$ cannot be charged below $\kappa$ if he gets allocated the good. Since $U_i(r'\mid a')=0$ it cannot be that $\kappa<r',$ as in that case the envelope theorem of \cite{milgrom_envelope_2002} and incentive-compatibility implies \[u_i(r'\mid a',m_{-i})=\int_\kappa^{r'}x_i(t\mid a',m_{-i})\,dt>0.\] Since $[r',r'+\epsilon]$ is not in the support of $F_{a'}$, it follows that $\kappa>r'+\epsilon.$ But then no type in $[r',r'+\epsilon]$ can get positive utility by reporting $a',$ which is a contradiction. 
\end{proof}

\begin{proof}[Proof of proposition \ref{prop:beliefs_support}.]
I first show that any type $v_i < r_i$ is eliminated by type $v_i' = r_i$, in the sense of definition \ref{def:refinement}. First note that, for $v_i$ to strictly prefer deviating, he must think the seller will form a posterior over his value that causes her to run the mechanism for some opponent message.

Consider any seller belief $G$ and let $\Gamma_G$ be the seller's optimal mechanism given $G$. Let $x_i(\cdot\mid G,m_{-i})$ and $t_i(\cdot\mid G,m_{-i})$ denote the seller's interim allocation and payment rules for bidder $i$ given the opponent $-i$ reported $m_{-i}$ and the seller's belief $G$ over bidder $i$'s value. Define
\[
\overline{x}_i(\hat v \mid G):= \E_{v_{-i}}\left[ x_i(\hat{v} \mid  G, \sigma_{-i}(v_{-i})) \right], \quad \overline{t}_i(\hat{v} \mid G) := \E_{v_{-i}}\left[ t_i(\hat{v} \mid G, \sigma_{-i}(v_{-i})) \right],
\]
accordingly. 

Now, consider the utility of type $v_i'=r_i$ under the same belief $G$. Then:
   \[ U_i(r_i \mid G) \ge r_i \cdot \overline{x}_i(v_i\mid G) - \overline{t}_i(v_i \mid G) =U_i(v_i \mid G) + (r_i - v_i) \cdot \overline{x}_i(v_i \mid G)\]

Now if $G$ is such that $U_i(v_i \mid G) \geq U_i(v_i \mid \sigma) = 0$ and $\overline{x}_i(v_i \mid G) > 0$, then
\[
U_i(r_i \mid G) - U_i(r_i \mid \sigma) = U_i(r_i \mid G) \ge U_i(v_i \mid G) + (r_i - v_i) \cdot \overline{x}_i(v_i \mid G)> U_i(v_i \mid G).\]
Thus, type $r_i$ has a strictly stronger incentive to deviate than any $v_i < r_i$, and the seller must therefore place zero probability on any $v_i \in [\underline{v}, r_i)$.
\end{proof}

\begin{proof}[Proof of lemma \ref{lem:exclusion}.]
Fix a nonactive message $n\in N_i$ and a positive-probability opponent message profile $m_{-i}$ for which the seller runs a mechanism after $(n,m_{-i})$. Let
\[
S_n:=\{v\in[\underline v,\overline v]:\sigma_i(v)=n\}.
\]
Since $n$ is on path and sent by a positive measure of types, $S_n$ has positive measure. By monotonicity of $\sigma_i$, $S_n$ is an interval up to endpoints.

By definition, $U_i(v\mid n)=0$ for every $v\in S_n$, and by IIR, $u_i(v\mid n,m'_{-i})\geq 0$ for every on-path opponent message profile $m'_{-i}$. Since $m_{-i}$ occurs with positive probability, it follows that
\[
u_i(v\mid n,m_{-i})=0
\qquad\text{for every }v\in S_n.
\]

By IC and the envelope representation, for any $v<v'$ in $S_n$,
\[
u_i(v'\mid n,m_{-i})-u_i(v\mid n,m_{-i})
=
\int_v^{v'}x_i(s\mid n,m_{-i})\,ds.
\]
The left side is zero for all $v<v'$. Since $x_i(\cdot\mid n,m_{-i})\geq 0$, this implies
\[
x_i(s\mid n,m_{-i})=0
\quad\text{a.e. }s\in S_n.
\]
Because $F_n$ is the prior conditional on $S_n$ and the prior is atomless with positive density, this is equivalent to
\[
x_i(s\mid n,m_{-i})=0
\quad\text{for }F_n\text{-a.e. }s.
\]
Thus bidder $i$ is excluded after $(n,m_{-i})$.
\end{proof}

\begin{proof}[Proof of lemma \ref{lem:pool_nonactive}.]
Fix such a $\sigma$ and a bidder $i$. Construct a new message space $\tilde M_i:=\{a_i,\bar n_i\}$, and define $\tilde\sigma_i$ as follows. For every type $v$, let $\tilde\sigma_i(v)$ agree with $\sigma_i(v)$ on messages in $a_i$, and assign to $\bar n_i$ the total probability that $\sigma_i(v)$ assigns to messages in $N_i$:
\[
\tilde\sigma_i(v)(a_i)=\sigma_i(v)(a_i),
\qquad
\tilde\sigma_i(v)(\bar n_i)=\sum_{n\in N_i}\sigma_i(v)(n).
\]
Define $\tilde\sigma_{-i}=\sigma_{-i}$ for the other bidders.

I claim that such a $\tilde\sigma$ is outcome-equivalent to $\sigma$. I will first establish the following.

\begin{enumerate}
\item\label{it:pp} If the seller runs a mechanism after any message of the form $(n_i,m_{-i})$, for some $n_i\in N_i$ and $m_{-i}\in M_{-i}$, then by Lemma \ref{lem:exclusion}, bidder $i$ is excluded almost surely. The seller's optimal mechanism therefore depends only on her posteriors about the remaining bidders. Furthermore, for every $n_i'\in N_i$, the seller runs this same mechanism after $(n_i',m_{-i})$.

\item\label{it:runA} If it was sequentially optimal for the seller to run the mechanism described in point \ref{it:pp} after any message of the form $(n_i,m_{-i})$, for some $n_i\in N_i$ and $m_{-i}\in M_{-i}$, then it is sequentially optimal for the seller to run the same mechanism under $(\bar n_i,m_{-i})$.

\item\label{it:run0} If $m_{-i}\in M_{-i}$ is such that the seller does not optimally run a mechanism for any message of the form $(n_i,m_{-i})$ where $n_i\in N_i$, then it is not optimal for the seller to run a mechanism under $(\bar n_i,m_{-i})$.
\end{enumerate}

I will now justify \ref{it:pp}--\ref{it:run0} and conclude outcome-equivalence. For \ref{it:pp}, by Lemma \ref{lem:exclusion}, bidder $i$ is excluded from any mechanism run after $(n_i,m_{-i})$ for $n_i\in N_i$. Thus, fixing $m_{-i}$, the seller's optimal mechanism depends only on her posteriors about the remaining bidders, and the same mechanism is optimal after any $(n_i',m_{-i})$ for $n_i'\in N_i$.

To establish \ref{it:runA} and \ref{it:run0}, fix $m_{-i}$ and let
\[
G_{-i}:=(F_{m_j})_{j\neq i}
\]
denote the vector of posteriors induced by the other bidders' messages. Under the pooled strategy $\tilde\sigma_i$, observing $\bar n_i$ is equivalent to the event $\{m_i\in N_i\}$ in the original equilibrium. By Bayes' rule, there exist weights $\alpha_{n_i}\geq 0$ with $\sum_{n_i\in N_i}\alpha_{n_i}=1$ such that
\[
F_{\bar n_i}=\sum_{n_i\in N_i}\alpha_{n_i}F_{n_i}.
\]

For any incentive-compatible mechanism $\Gamma=(x,t)$ the seller may run, its expected revenue under independent posteriors $F_i$ and $G_{-i}$ is
\[
\Rev(\Gamma;F_i,G_{-i})
=
\E_{v_i\sim F_i,\,v_{-i}\sim G_{-i}}
\left[\sum_{j=1}^n t_j(v)\right].
\]
For every fixed $\Gamma$, this expression is linear in $F_i$. Therefore, if
\[
\Rev(F_i,G_{-i})
:=
\sup_{\Gamma\in\mathcal G}\Rev(\Gamma;F_i,G_{-i})
\]
denotes the seller's optimal revenue, then $\Rev(F_i,G_{-i})$ is convex in $F_i$. In particular,
\[
\Rev(F_{\bar n_i},G_{-i})
\leq
\sum_{n_i\in N_i}\alpha_{n_i}\Rev(F_{n_i},G_{-i}).
\]

For \ref{it:run0}, if the seller does not optimally run after $(n_i,m_{-i})$ for any $n_i\in N_i$, then
\[
\Rev(F_{n_i},G_{-i})\leq c
\qquad\text{for all }n_i\in N_i.
\]
Hence
\[
\Rev(F_{\bar n_i},G_{-i})
\leq
\sum_{n_i\in N_i}\alpha_{n_i}\Rev(F_{n_i},G_{-i})
\leq c,
\]
so it is sequentially optimal for the seller not to run after $(\bar n_i,m_{-i})$.

For \ref{it:runA}, suppose instead that the seller runs after some $(n_i,m_{-i})$. By \ref{it:pp}, the same mechanism, up to behavior on zero-probability type profiles, is optimal after every $(n_i',m_{-i})$ for $n_i'\in N_i$, and it excludes bidder $i$ almost surely. Let $R(m_{-i})$ denote its revenue. Since bidder $i$ is excluded a.s., $R(m_{-i})$ does not depend on which posterior $F_{n_i}$ the seller has about bidder $i$, and
\[
\Rev(F_{n_i},G_{-i})=R(m_{-i})
\qquad\text{for all }n_i\in N_i.
\]
Convexity implies
\[
\Rev(F_{\bar n_i},G_{-i})
\leq
\sum_{n_i\in N_i}\alpha_{n_i}\Rev(F_{n_i},G_{-i})
=
R(m_{-i}).
\]
But the same mechanism that excludes bidder $i$ is feasible after $(\bar n_i,m_{-i})$ and yields revenue $R(m_{-i})$. Therefore it is optimal after $(\bar n_i,m_{-i})$ as well. This establishes \ref{it:runA}.

We now define the seller's strategy $\tilde\pi$ in the pooled-message game. After any profile $(a_i,m_{-i})$, let $\tilde\pi$ coincide with the original equilibrium continuation after $(a_i,m_{-i})$ under $\sigma$. After any profile $(\bar n_i,m_{-i})$, let $\tilde\pi$ run the mechanism described in \ref{it:pp} if the seller ran after $(n_i,m_{-i})$ for some, and hence every, $n_i\in N_i$ under $\sigma$, and otherwise let $\tilde\pi$ take the outside option. By \ref{it:pp}--\ref{it:run0}, this strategy is sequentially optimal and yields the same continuation outcome as the original equilibrium conditional on the event $\{m_i\in N_i\}$.

Since $\tilde\pi$ matches the original continuation after all profiles with $m_i=a_i$ and matches the original continuation conditional on $\{m_i\in N_i\}$ after profiles with $\bar n_i$, the induced distribution over allocations and payments is identical under $(\tilde\sigma,\tilde\pi)$ and $(\sigma,\pi)$. Thus the former is an equilibrium and the two are outcome-equivalent. Pooling the messages in $N_j$ for each other bidder $j\neq i$ in an analogous way proves the lemma.
\end{proof}

\begin{proof}[Proof of lemma \ref{lem_active_nonactive}.]
This proof uses the following standard facts about the general Myerson optimal auction when bidders have independent distributions on a connected set. Since bidder strategies are monotone, every on-path posterior is the prior conditional on an interval, possibly open at either endpoint. For a posterior $G$, let $\bar\phi_G$ denote the (potentially) ironed virtual type function with respect to $G$. The Myerson allocation rule maximizes ironed virtual surplus pointwise.

If $b=\sup(\supp(G))$, then the following holds
\[
    \bar\phi_G(v)\to b
    \qquad\text{as }v\uparrow b\text{ along }\supp(G).
\]
Moreover, if $y<b$, then ironed virtual values just above $y$ are bounded strictly below $y$: there exist $\eta,\delta>0$ such that
\[
    \bar\phi_G(v)<y-\eta
    \qquad
    \text{for all }v\in(y,y+\delta)\cap\supp(G).
\]

Suppose the seller runs a mechanism after some message profile in which at least one bidder sends his pooled nonactive message. Fix one such bidder $i$. Since the active posterior FOSDs the nonactive posterior, the seller must also run after $(\bar n_i,a_{-i})$, where bidder $i$ sends his nonactive message and every other bidder sends his active message \citep{hart_reny_15}.

By Lemma \ref{lem:exclusion}, bidder $i$ must be excluded from any mechanism run after $(\bar n_i,a_{-i})$. Since types $v>r$ obtain strictly positive utility from the active message and zero utility from the nonactive message, no type above $r$ sends $\bar n_i$. Hence
\[
F_{\bar n_i}(r)=1.
\]
I claim there exists $\epsilon>0$ such that
\[
F_{\bar n_i}(r-\epsilon)=1.
\]
If for sake of contradiction this were not the case, then for every $\epsilon>0$, it must be that $F_{\bar n_i}(r-\epsilon)<1$, so
\[
\sup\{\supp(F_{\bar n_i})\}=r.
\]
I will show that gives a contradiction. Let $\bar\phi_i(\cdot\mid \bar n_i)$ and $\bar\phi_j(\cdot\mid a_j)$ denote the (potentially) ironed\footnote{Monotonicity of strategies implies the posteriors are supported on connected sets, so the standard Myerson ironing procedure goes through.} virtual value functions corresponding to $F_{\bar n_i}$ and $F_{a_j}$ for $j\neq i$, respectively. It is well-known from standard mechanism design that the ironed virtual value of bidder $i$ along the upper tail of $F_{\bar n_i}$ converges to the top value, i.e.
\begin{equation}\label{eqn:lower_virtual}
    \bar\phi_i(v\mid \bar n_i)\to r
    \quad \text{as } v\uparrow r \text{ along } \supp(F_{\bar n_i}).
\end{equation}

Furthermore, since the set $(r,r+\epsilon')\cap\supp(F_{a_j})$ cannot be empty for any $\epsilon'>0$ by the same argument as in the proof of Lemma \ref{lem:common_reserve}, there exist $\epsilon,\epsilon'>0$ such that
\begin{equation*}
\bar\phi_j(v\mid a_j)<r-\epsilon
\qquad\text{for all }j\neq i
\quad\text{and all }v\in(r,r+\epsilon')\cap\supp(F_{a_j}).
\end{equation*}
By \eqref{eqn:lower_virtual}, there exists $\delta>0$ such that
\[
\bar\phi_i(v\mid \bar n_i)>r-\epsilon
\qquad\text{for all }v\in(r-\delta,r)\cap\supp(F_{\bar n_i}).
\]
Now there is strictly positive probability on the event
\[
v_i\in(r-\delta,r)
\quad\text{and}\quad
v_j\in(r,r+\epsilon')\text{ for all }j\neq i.
\]
On this event,
\[
\bar\phi_i(v_i\mid \bar n_i)>r-\epsilon>
\max_{j\neq i}\bar\phi_j(v_j\mid a_j),
\]
and $\bar\phi_i(v_i\mid \bar n_i)\geq 0$. The Myerson allocation rule therefore allocates the good to bidder $i$ on this event, implying bidder $i$ wins with strictly positive probability after $(\bar n_i,a_{-i})$. This contradicts Lemma \ref{lem:exclusion}. Hence there exists $\epsilon>0$ such that
\[
F_{\bar n_i}(r-\epsilon)=1.
\]

Since the prior $F$ is full support on $[\underline v,\overline v]$ and all nonactive messages have been pooled, this implies that $F_{a_i}(r)>0$. By symmetry, the same is true for every active bidder $j\neq i$. But then, after $(\bar n_i,a_{-i})$, there is positive probability that every active bidder has value below $r$. On this event, none of the bidders using their active message can get the good by the definition of $r.$ If bidder $i$ gets the good with positive probability, this contradicts Lemma \ref{lem:exclusion}. If not, the seller must be withholding the good on a positive probability event where bidder $i$ has a positive virtual type but all other bidders have a value below $r,$ contradicting seller optimality. Therefore the seller cannot run a mechanism after any message profile in which some bidder sends his nonactive message.
\end{proof}

\begin{proof}[Proof of lemma \ref{lem:r_pF}.]
Fix such an equilibrium. Let $G$ denote the seller's posterior distribution of a bidder's value conditional on receiving the active message $a_i$, and let $q:=\Pr(m_i=a_i)$. By symmetry, this posterior is the same for every bidder after the all-active profile.

Because every type $v\geq r$ sends the active message, for any $p\geq r$,
\[
1-G(p)=\Pr(v_i\geq p\mid m_i=a_i)=\frac{1-F(p)}{q}.
\]
After the all-active profile $(a_1,\dots,a_n)$, the seller faces $n$ symmetric bidders with posterior $G$, so the reserve $r$ must maximize $p(1-G(p))$. But on $[r,\overline v]$, this objective is equal to $\frac{1}{q}p(1-F(p))$. If $r<p_F$, then $p_F$ is feasible and, by uniqueness of the monopoly price under $F$, yields strictly larger revenue than $r$. This contradicts optimality of $r$ under $G$. Therefore $r\geq p_F$.
\end{proof}

\begin{proof}[Proof of lemma \ref{lem:restricted_threshold_wlog}.]
Let $A_i$ and $N_i$ be the active and nonactive messages defined in Section \ref{unrestricted}. If either set is empty, the equilibrium is already binary up to a degenerate threshold, so assume both are nonempty. Lemmas \ref{lem:common_reserve} and \ref{lem:one_active} use only bidder incentives and monotonicity, so they apply here and imply that there is at most one active message, denoted $a$. It remains to check that Lemma \ref{lem:exclusion} is also regime-independent.

Since $a$ gives some type positive utility, every higher type strictly prefers $a$ to any nonactive message, which gives zero utility. Monotonicity of bidder strategies then implies that the types sending $a$ must include an upper interval $(\tau,\overline v]$ for some $\tau$, and no type strictly below $\tau$ sends $a$. Thus all types $v<\tau$ must be sending a nonactive message. If an auction with a common reserve runs on-path and allocates the good to $v$ with positive probability, the common reserve $r$ lies strictly below $\tau.$ This is a contradiction, since any optimal common reserve in an auction where at least one bidder uses $a$ must be at least $\tau.$
\end{proof}

\begin{proof}[Proof of proposition \ref{prop:restricted_char}.]
    In any symmetric threshold equilibrium with threshold $\tau$, the posterior distribution of a bidder sending $H$ is the truncation of $F$ to $[\tau, \overline{v}]$, and for $L$ it is the truncation to $[\underline{v}, \tau)$.
    \begin{enumerate}
        \item Suppose the auction runs and the profile is $(H, H)$, so both values are in $[\tau, \overline{v}]$. Since the optimal reserve is the optimal monopoly price, $r = p_F$ if $\tau\leq p_F$ and $r=\tau$ if $\tau>p_F.$
        \item If the profile is $(H, L)$ and the auction is run, an optimal common reserve cannot be below $\tau$ since the bidder sending $H$ has value at least $\tau$, so the problem again reduces to the optimal posted price for the bidder sending $H$, which is $r = p_F$ if $\tau\leq p_F$ and $r=\tau$ if $\tau>p_F.$
        \item The auction cannot run under $(L,L)$ if it is run under $(H,H)$, since this would induce a strictly lower reserve price than the $r$ set according to the above two points, which would induce types weakly above $r$ to deviate to $L.$ This unravels the equilibrium in the same manner as lemma \ref{lem:common_reserve}.
    \end{enumerate} 
This proves the proposition.
\end{proof}
\begin{proof}[Proof of Proposition \ref{prop:entry_costs}.]
Fix a baseline threshold equilibrium with cutoff $\tau$, and let $H$ denote the message sent by types weakly above $\tau$. In the perturbed game, consider the strategy
\[\sigma_\epsilon(v)=
\begin{cases}
H & \text{if }v\geq\tau,\\
\varnothing & \text{if }v<\tau.
\end{cases}
\]
A bidder who sends $H$ posts a deposit $\epsilon$. A bidder who sends $\varnothing$ posts no deposit and is not eligible to participate in any mechanism the seller runs. If the seller does not run a mechanism, deposits are refunded. Any off-path costly message is assigned the same belief and continuation treatment as $H$.

First observe that the seller's continuation problem is unchanged in the perturbed game compared to the baseline model, conditional on the set of participating bidders. Indeed, fix a set of bidders sending $H$ (call them participating bidders) and an ordinary IIR mechanism $(x,t)$ for those bidders. In $\Gamma_\epsilon$, the seller can use the same allocation rule and set the mechanism-stage transfer of each participating bidder equal to $t_i-\epsilon$. The seller can keep the deposits $\epsilon$ whenever she runs, so her net payoff is unchanged. Each participating bidder's mechanism-stage utility is increased by $\epsilon$, so the modified IIR constraint is satisfied. Conversely, given any allowable mechanism in $\Gamma_\epsilon$, adding the retained deposit to each depositing bidder's mechanism-stage transfer gives an ordinary IIR mechanism with the same allocation and the same seller payoff. Thus, conditional on the set of participating bidders, the seller's optimal net continuation value is the ordinary optimal revenue from those bidders. From here it is straightforward to check that the seller makes the same continuation decisions as in the baseline threshold equilibrium.

It remains to compare bidder incentives. First suppose the baseline equilibrium is a restricted-regime equilibrium in which the seller runs whenever at least one bidder reports $H$, and suppose $\tau\geq p_F$. If exactly one bidder sends $H$, the seller runs a second-price auction for the sole entrant with reserve
\[
    \rho_\epsilon=\tau-\epsilon.
\]
Since the high posterior is supported on
$[\tau,\overline v]$, the allocation rule is the same as in the baseline mechanism with reserve $\tau$. The cutoff type obtains mechanism-stage utility $\epsilon$, and after the retained deposit is counted obtains total utility zero. The seller's total payment from any entrant who buys is $\rho_\epsilon+\epsilon=\tau$, exactly as in the baseline equilibrium.

If both bidders send $H$, lowering the reserve below $\tau$ does not give the cutoff type positive allocation surplus, since the opponent is drawn from a truncation of the prior to $[\tau,\overline v]$. The seller therefore uses the baseline common-reserve allocation and gives each entrant a fixed rebate $\epsilon$ at the mechanism stage. This rebate is exactly offset by the retained deposit, so the seller's net revenue is unchanged and every depositing type satisfies modified IIR.

Under this construction, every type $v\geq\tau$ obtains the same total continuation utility from sending $H$ as in the baseline equilibrium, and is therefore willing to send $H$. A type $v<\tau$ obtains zero from $\varnothing$. If he deviates to $H$ and the other bidder sends $H$, he can
collect the rebate and choose a report that loses, giving total utility zero. If the other bidder sends $\varnothing$, the deviator is the sole entrant and faces total price $\tau$ after the reserve concession and retained deposit are counted. His best payoff in that state is
\[
    \max\{v-(\tau-\epsilon),0\}-\epsilon
    =
    \max\{v-\tau,-\epsilon\}<0.
\]
Since the event that the other bidder sends $\varnothing$ has positive probability, every type $v<\tau$ strictly prefers $\varnothing$.

Next suppose the baseline equilibrium is a restricted one-high equilibrium with $\tau<p_F$. At every run profile, the baseline common reserve is $p_F$. Since $p_F>\tau$, a small reserve concession does not give the cutoff type positive
surplus. The seller therefore uses the baseline allocation and gives each depositing bidder a fixed rebate $\epsilon$ at the mechanism stage. The retained deposit exactly offsets this rebate, so every type $v\geq\tau$ obtains the same total utility from sending $H$ as in the baseline. A type $v<\tau$ who deviates to $H$ can collect the rebate whenever the seller runs and otherwise has his deposit refunded. Since $v<\tau<p_F$, he obtains no positive allocation surplus. His deviation payoff is therefore zero, so he is only weakly willing to use
$\varnothing$.

Finally suppose the baseline equilibrium is either a restricted-regime equilibrium in which the seller runs only after both bidders report $H$, or an unrestricted threshold equilibrium. In both cases the seller does not run after $(H,\varnothing)$, so a deviating bidder's deposit is refunded in that state. At $(H,H)$, by similar reasoning as above, the seller must use the baseline allocation and give each depositing bidder a fixed rebate $\epsilon$ at the mechanism stage. Therefore, while the equilibrium outcomes and utilities match those of the baseline model, the indifference of low types remains.
\end{proof}

\begin{proof}[Proof of proposition \ref{prop:lying_costs_restricted}.]
Fix $\epsilon$. Construct the strategy profile $\sigma_\epsilon$ as follows: types $v < \tau$ report truthfully ($m = v$), and types $v \ge \tau$ pool at message $m = \tau$. On path, for messages $m < \tau$, the seller forms a point mass belief at $m$, leaving the bidder with type $m$ zero interim utility. For $m = \tau$, the seller forms the truncated belief $F|_{[\tau, \overline{v}]}$. For any off-path message $m >\tau$ sent by bidder $i$, the seller has a point-mass belief on $\overline{v}$ about bidder $i$'s value. I claim that this off-path belief satisfies the reasonable seller beliefs condition in Definition \ref{def:refinement}. Indeed, if bidder $i$ believes that by sending an off-path message $m>\tau$, the seller will have the belief $F\mid_{v\geq \tau}$ about his valuation, then all types $v>\tau$ have a strict incentive to deviate to avoid the lying cost. Since $\ell$ is strictly submodular, $\overline v$ has the largest incentive to deviate under this candidate off-path belief. Thus $\overline{v}$ can never be eliminated by the reasonable seller beliefs requirement.

It is easy to see that when $c\geq \tau$, the seller's behavior if bidders report according to $\sigma_\epsilon$ in the perturbed model matches the seller's behavior under $\sigma$ in the baseline model for all pairs of bidder valuations. For a type $v < \tau$, under truthful reporting, his equilibrium payoff is $0$. If he deviates to the threshold message $m=\tau,$ his interim utility from the mechanism the seller runs is still 0, but he incurs a strictly positive lying cost $\epsilon \ell(\tau, v) > 0$. Thus his net payoff from this deviation is strictly negative, so truthful reporting is strictly optimal.

It remains to check bidder incentives for types $v \ge \tau$. Sending the pooled message $m = \tau$ yields an expected mechanism utility $U_i(v \mid \tau) \ge 0$, and incurs lying cost $\epsilon \ell(\tau, v)$. Deviating to a truthful message $m = v > \tau$ yields an interim utility of $0$ from the mechanism and $0$ lying cost. The net change in total utility from playing the equilibrium strategy rather than deviating is $U_i(v \mid \tau) - \epsilon \ell(\tau, v)$. The remainder of the proof shows this difference is weakly positive for all $v \geq \tau$.

By the envelope theorem, \[U_i(v \mid \sigma) = \int_\tau^v \bar{x}_i(s \mid \tau) ds.\]
Because the auction runs at least when both bidders report $H$, a bidder with type $s \ge \tau$ wins whenever the opponent also reports $H$ but has a lower true type. Thus, his expected allocation probability is bounded below by the probability the opponent's type falls in $[\tau, s]$:
\[
\bar{x}_i(s \mid \tau) \geq F(s) - F(\tau).
\]
Since the prior density $f$ is strictly positive and continuous on the compact support $[\underline{v}, \overline{v}]$, it has a minimum value $\delta = \min_{v \in [\underline v,\overline v]} f(v) > 0$. Thus the allocation rule can be bounded by $\bar{x}_i(s \mid \tau) \ge \delta(s - \tau)$, implying
\[
U_i(v \mid \sigma) \ge \int_\tau^v \delta(s - \tau) ds = \frac{\delta}{2}(v - \tau)^2.
\]

To bound the lying cost, since $\ell(\tau, v)$ is twice continuously differentiable on a compact domain, a Taylor expansion is valid around $v = \tau$:
\[
\ell(\tau, v) = \ell(\tau, \tau) + \ell_2(\tau, \tau)(v - \tau) + \frac{1}{2}\ell_{22}(\tau, \tilde{v})(v - \tau)^2.
\]
for some $\tilde{v} \in (\tau, v)$. The assumptions on $\ell$ imply that $\ell(\tau, \tau) = 0$ and $\ell_2(\tau, \tau) = 0$. Let $M = \max_{v, m \in [\underline v,\overline v]} |\ell_{22}(m, v) |$, which is finite because the type space is compact. Then
\[
\ell(\tau, v) \le \frac{M}{2}(v - \tau)^2.
\]

Comparing the two bounds, the equilibrium strategy is strictly optimal for all types $v > \tau$ as long as
\[
\frac{\delta}{2}(v - \tau)^2 \ge \epsilon \frac{M}{2}(v - \tau)^2\Leftrightarrow \epsilon \le \frac{\delta}{M}.
\]

Thus, picking $\bar{\epsilon}< \frac{\delta}{M}$, for any $\epsilon < \bar{\epsilon}$, the strategy $\sigma_\epsilon$ is an equilibrium. The result follows.
\end{proof}

\begin{proof}[Proof of proposition \ref{prop:lying_costs_unrestricted}.]
The proof is similar to that of the previous proposition, except it needs to be shown that the additional restriction imposed by the additional condition (\ref{eq:lying_unrestricted_condition}) ensures that the seller does not run a mechanism when one bidder discloses a type $v<\tau$ and the other sends message $\tau.$ Indeed, if the seller knows that bidder 1 has value $v_1<\tau$ for sure while bidder 2 reports a value $v_2\geq \tau$, then the seller allocates to bidder 1 as long as $\varphi(v_2)<v_1$ and allocates to bidder 2 otherwise. Thus her expected optimal auction revenue is \[\E[\max(v_1,\varphi(v_2))\mid v_2\geq \tau].\] If the seller does not run the mechanism at such a message profile, it must be that \[c\geq \E[\max(v_1,\varphi(v_2))\mid v_2\geq \tau].\] Since the above must hold for all $v_1<\tau$, the condition is equivalent to \[c\geq \E[\max(\tau,\varphi(v))\mid v\geq \tau].\]
\end{proof}

\begin{lemma}[Off-path belief verification]
\label{lem:offpath_threshold_beliefs}
Consider the outcome-equivalent threshold equilibrium described in Theorem \ref{myerson_characterization} where the seller runs iff all bidders are above the threshold, and the on-path auction uses reserve $r$. Suppose that after any off-path message from bidder $i$, the seller's belief about bidder $i$ is the truncation of $F$ to $[r,\overline v]$. Then these off-path beliefs satisfy Definition \ref{def:refinement}. Moreover, no bidder type has a profitable deviation to an off-path message.
\end{lemma}

\begin{proof}
By Proposition \ref{prop:beliefs_support}, Definition \ref{def:refinement} rules out all types below $r$ after any off-path message, and the belief $F\mid_{[r,\overline v]}$ satisfies the refinement.

Now consider bidder incentives. A type $v_i<r$ cannot obtain positive utility from an off-path message because any mechanism chosen after the belief $F\mid_{[r,\overline v]}$ gives bidder $i$ a report space contained in $[r,\overline v]$, and the resulting individualized reserve is at least $r$. Thus every feasible report gives type $v_i<r$ weakly negative surplus, and the bidder can always obtain zero by not winning the object.

A type $v_i\geq r$ already sends the active message. Under the active-path mechanism, reports below $r$ do not win the object, so replacing the active posterior by $F\mid_{[r,\overline v]}$ does not create a more favorable continuation problem for bidder $i$. Hence, fixing the other bidders' messages, no type $v_i\geq r$ can obtain a strictly higher payoff from such an off-path message than from the active message. Therefore no bidder type has a profitable off-path deviation.
\end{proof}

\section{Examples of non-monotone equilibria}\label{appendix:nonmonotone_examples}

This appendix gives two examples showing that the monotonicity restriction can rule out symmetric equilibrium outcomes in a finite-type analogue of the model. The examples do not satisfy the atomless-type assumption maintained in the main text but serve to build intuition for how meaningfully different equilibria can arise once non-monotone mixed strategies are allowed. While the existence of such equilibria is of theoretical interest, the constructions of these equilibria rely on the fact that the seller’s optimal mechanism under either regime is highly sensitive to her posteriors about the bidders’ valuations, and therefore bidder types must be using mixed strategies at the exact probabilities that make them indifferent between two or more seller best-responses. As the analysis that follows will make clear, these equilibria seem quite knife-edge and do not seem very robust or suitable for applications.

In both examples, there are two bidders, values are independently drawn from a finite type space, the first-period message space is $\{\ell,m,m'\}$, and both bidders use the same strategy. Since all three messages are sent with positive probability, every message profile is on-path. Consistent with the main text, when a deviating bidder's value is outside the posterior support induced by his message, his payoff is computed from the best report in that support.

\subsection{Unrestricted seller}

First consider the unrestricted regime. Let values be drawn from
\[
\begin{array}{c|cccc}
v & \frac12 & 1 & 2 & 3\\
\hline
\Pr(v) & \frac{1}{5} & \frac{4}{15} & \frac{2}{5} & \frac{2}{15}.
\end{array}
\]
Let the seller's outside option be
\[
c=1.9.
\]
Consider the symmetric strategy
\[
\Pr(\ell\mid \tfrac12)=1,
\]
\[
\Pr(m\mid v)=
\begin{cases}
\frac{1}{2} & \text{if } v=1,\\
\frac{1}{3} & \text{if } v=2,\\
1 & \text{if } v=3,
\end{cases}
\qquad
\Pr(m'\mid v)=1-\Pr(m\mid v)
\quad\text{for }v\in\{1,2,3\}.
\]
Thus type $\frac12$ sends $\ell$, types $1$ and $2$ mix between $m$ and $m'$, and type $3$ sends $m$. This strategy is not monotone.

The seller's posteriors after $\ell,m,$ and $m'$ are
\[
F_\ell\left(\frac12\right)=1,
\]
\[
\begin{array}{c|ccc}
v & 1 & 2 & 3\\
\hline
\Pr(v\mid m) & \frac{1}{3} & \frac{1}{3} & \frac{1}{3},
\end{array}
\qquad
\begin{array}{c|cc}
v & 1 & 2\\
\hline
\Pr(v\mid m') & \frac{1}{3} & \frac{2}{3},
\end{array}
\]
respectively.
The virtual types according to the posteriors induced by $m$ and $m'$ are
\[
\begin{array}{c|ccc}
v & 1 & 2 & 3\\
\hline
\varphi_m(v) & -1 & 1 & 3,
\end{array}
\qquad
\begin{array}{c|cc}
v & 1 & 2\\
\hline
\varphi_{m'}(v) & -1 & 2.
\end{array}
\]
The posterior after $\ell$ is degenerate at $\frac12$, so the corresponding virtual value is $\frac12$. These virtual type functions are increasing.

The seller's sequentially-optimal actions are the following:
\[
\begin{array}{c|c|c}
\text{message profile} & \text{optimal revenue} & \text{seller action}\\
\hline
(\ell,\ell) & \frac12 & \varnothing\\
(\ell,m) & \frac32 & \varnothing\\
(\ell,m') & \frac32 & \varnothing\\
(m,m) & 2 & \text{run auction}\\
(m,m') & 2 & \text{run auction}\\
(m',m') & \frac{16}{9} & \varnothing.
\end{array}
\]

It remains to check bidder incentives and see where this example differs from a threshold equilibrium. In particular, in this example, two different mechanisms run on-path. Under $(m,m)$, both bidders have report space $\{1,2,3\}$ and virtual types $-1,1,3$, so the induced mechanism is a second-price auction with reserve 2. The reduced-form allocation and payment rules are
\[
(x,t)(1)=(0,0),\qquad
(x,t)(2)=\left(\frac12,1\right),\qquad
(x,t)(3)=\left(\frac56,2\right).
\]
Now consider the mechanism under $(m,m')$. The bidder who sent $m$ has possible virtual types $-1,1,3$, while the bidder who sent $m'$ has possible virtual types $-1,2$. Hence an $m'$-sender with type $2$ wins against an $m$-sender reporting $2$ deterministically. The corresponding rules are
\[
(x_m,t_m)(1)=(0,0),\qquad
(x_m,t_m)(2)=\left(\frac13,\frac23\right),\qquad
(x_m,t_m)(3)=\left(1,\frac83\right),
\]
for the bidder who sent $m$, and
\[
(x_{m'},t_{m'})(1)=(0,0),\qquad
(x_{m'},t_{m'})(2)=\left(\frac23,\frac43\right),
\]
for the bidder who sent $m'$. Thus two different mechanisms run on-path.

The resulting expected utilities are
\[
\begin{array}{c|ccc}
v & U(v\mid \ell) & U(v\mid m) & U(v\mid m')\\
\hline
\frac12 & 0 & 0 & 0\\
1 & 0 & 0 & 0\\
2 & 0 & 0 & 0\\
3 & 0 & \frac13 & \frac{4}{15}.
\end{array}
\]
From here it is easy to see that types 1 and 2 are willing to mix, while type 3 strictly prefers $m.$

\subsection{Restricted seller}

Now consider the restricted regime, where the seller can only run a second-price auction with a common reserve. Let values be drawn from
\[
\begin{array}{c|cccc}
v & 9 & 10 & 11 & 12\\
\hline
\Pr(v) & \frac{1}{100} & \frac{3}{20} & \frac{9}{25} & \frac{12}{25}.
\end{array}
\]
Let the seller's outside option be
\[
c=9.5.
\]
Consider the bidder-symmetric strategies
\[
\Pr(m\mid v)=
\begin{cases}
\frac{2}{5} & \text{if } v=10,\\
1 & \text{if } v=12,\\
0 & \text{if } v\in\{9,11\},
\end{cases}
\]
\[
\Pr(m'\mid v)=
\begin{cases}
\frac{3}{5} & \text{if } v=10,\\
1 & \text{if } v=11,\\
0 & \text{if } v\in\{9,12\},
\end{cases}
\qquad
\Pr(\ell\mid 9)=1.
\]
Thus type $9$ sends $\ell$, type $10$ mixes between $m$ and $m'$, type $11$ sends $m'$, and type $12$ sends $m$. The strategy is not monotone.

Thus the seller's posteriors after $\ell,m,$ and $m'$ are as follows:
\[
F_\ell(9)=1,
\]
\[
\begin{array}{c|cc}
v & 10 & 12\\
\hline
\Pr(v\mid m) & \frac{1}{9} & \frac{8}{9},
\end{array}
\qquad
\begin{array}{c|cc}
v & 10 & 11\\
\hline
\Pr(v\mid m') & \frac{1}{5} & \frac{4}{5}.
\end{array}
\]

The following table gives one optimal common reserve for each message profile:
\[
\begin{array}{c|c|c|c}
\text{message profile} & \text{common reserve} & \text{revenue} & \text{seller action}\\
\hline
(\ell,\ell) & 9 & 9 & \varnothing\\
(\ell,m) & 12 & \frac{32}{3} & \text{run auction}\\
(\ell,m') & 10 & 10 & \text{run auction}\\
(m,m) & 12 & \frac{320}{27} & \text{run auction}\\
(m,m') & 11 & \frac{484}{45} & \text{run auction}\\
(m',m') & 10 & \frac{266}{25} & \text{run auction}.
\end{array}
\]

It is easy to check that the above reserve prices and seller actions are sequentially optimal for the seller.

Again note that this equilibrium cannot be outcome-equivalent to the threshold equilibria analyzed in the main text since more than one kind of mechanism runs on-path. It remains to check bidder incentives. Let $U(v\mid a)$ denote the expected utility of a bidder with value $v$ from sending message $a$, averaging over the other bidder's first-period message. The induced utilities are
\[
\begin{array}{c|ccc}
v & U(v\mid \ell) & U(v\mid m) & U(v\mid m')\\
\hline
9 & 0 & 0 & 0\\
10 & 0 & 0 & 0\\
11 & 0 & 0 & \frac{1}{10}\\
12 & 0 & \frac{9}{20} & \frac{11}{25}.
\end{array}
\]
Hence type $9$ is willing to send $\ell$, type $10$ is indifferent between $m$ and $m'$, type $11$ strictly prefers $m'$, and type $12$ strictly prefers $m$. Therefore this strategy, together with the seller behavior in the table, is a bidder-symmetric equilibrium.
\begin{singlespace}
\bibliographystyle{aer}
\bibliography{sources}
\end{singlespace}
\end{document}